\documentclass[aps,showpacs,amsmath,amssymb,twocolumn,prx,longbibliography,superscriptaddress,notitlepage,floatfix,10pt]{revtex4-2}

\usepackage[dvips]{graphicx}
\usepackage{bm, bbm,ascmac,amsmath,amssymb,amsthm,mathrsfs,amsfonts,dsfont}
\usepackage{epsfig}
\usepackage{braket}
\usepackage{enumerate}
\usepackage{xcolor}
\usepackage{float}
\usepackage{algorithm}
\usepackage{algorithmic}
\usepackage{comment}
\usepackage{appendix}
\usepackage{here}
\usepackage{tabularx}
\usepackage{dcolumn}
\usepackage[caption=false]{subfig}
\usepackage{adjustbox}
\usepackage[hidelinks]{hyperref}
\usepackage{bookmark}
\usepackage{makecell}

\newtheorem{theorem}{Theorem} 
\newtheorem{lemma}{Lemma} 
\newtheorem{corollary}{Corollary} 

\usepackage{scalerel}

\iftrue
  \usepackage{marginnote}
\fi

\begin{document}


\title{Beyond Commutativity: Redesigning Trotter Decomposition via Local Symmetry}


\author{Naoki Negishi}
\email{negishi@roma2.infn.it}
\affiliation{Dipartimento di Fisica, Università di Roma Tor Vergata, Via della Ricerca Scientifica 1, 00133 Rome, Italy}
\affiliation{INFN, Sezione di Roma Tor Vergata, Via della Ricerca Scientifica 1, 00133 Rome, Italy}

\author{Bo Yang}
\email{yangbo@g.ecc.u-tokyo.ac.jp}
\affiliation{Graduate School of Information Science and Technology, The University of Tokyo, Bunkyo-ku, Tokyo 113-8656, Japan}
\affiliation{LIP6, Sorbonne Université, CNRS, 4 place Jussieu, 75005 Paris, France}
\thanks{N.N. and B.Y. contributed equally to this work.}

\begin{abstract}
The product formula, commonly known as Trotter decomposition, is a central tool for digital quantum simulation, whose performance depends critically on how the Hamiltonian is partitioned into tractable blocks.
Standard decompositions typically rely on direct commutativity among Hamiltonian terms in a chosen operator representation, which can lead to large residual errors and deep circuits for complex, practically relevant many-body quantum systems.
We address this fundamental bottleneck by introducing a new decomposition principle that goes beyond commutativity, grouping Hamiltonian terms into local three-site clusters according to the underlying $\mathrm{SU}(2)$ symmetry of the local dynamics.
We show that three-site generators fall into at most four $\mathrm{SU}(2)$-symmetry classes, each admitting an effective two-qubit $\mathrm{SU}(4)$ representation with exact and efficient implementations.
By reducing the number of clusters, this decomposition principle substantially suppresses commutator-induced errors and circuit overhead while preserving underlying physical structures that commutativity-based decompositions may violate.
We demonstrate the proposed method on several physically relevant spin-lattice models, where the reduced cluster structure can even realise the second-order product formula without doubling the circuit depth, as would be required by conventional decompositions.
Numerical simulations of a Kagome Heisenberg model with triangular spin-chirality interactions show that the proposed method reduces both state infidelity and average spin-chirality bias by more than three orders of magnitude compared with conventional decompositions, while using substantially fewer gates.
These results establish local symmetry as a flexible and practical design principle for product-formula simulation, opening a route to more accurate and hardware-efficient simulations of broader classes of many-body systems.
\end{abstract}

\maketitle


\section{Introduction}

Simulating many-body quantum dynamics via product formula, known as Trotter decomposition~\cite{Trotter1959On, Suzuki1976Generalized, Hatano2005Finding}, is a long-standing, foundational tool in studying condensed matter physics and quantum chemistry.
The Trotter decomposition approximates the system's time evolution by sequentially applying short-time propagators generated by tractable, non-commuting components of the underlying Hamiltonian.
This product-formula structure is particularly well-suited to digital quantum simulation on quantum hardware, as it can be implemented via repeated applications of local gate blocks in a quantum circuit.
This makes Trotter decomposition a core building block for exploring near-term quantum utility in simulating large-scale many-body systems that are challenging for classical devices.

However, a substantial gap remains between current milestone demonstrations on quantum hardware~\cite{Kim2023Evidence, Whitlow2023Quantum, Chowdhury2024Enhancing, Farrell2024Quantum, Miessen2024Benchmarking, Yoshioka2025Krylov} and the simulation of more intricate, material-specific Hamiltonians, ranging from frustrated many-body systems with complex lattice geometry~\cite{Di-Sante2026Kagome} to molecular catalysts with dense electronic-structure Hamiltonians~\cite{Luo2025Efficient, Kan2025Resource-optimized, Ronevi2026A}.
For many-body Hamiltonians, frustrated geometries and anisotropic or long-range interactions can significantly increase both the number of propagators required per Trotter step and the commutator-induced residual error.
Since the decomposition determines both the propagator sequence and the non-commuting operator pairs, the cost and accuracy of Trotter-based simulation depend sensitively on how the Hamiltonian is decomposed.
Suppressing the resulting error generally requires more Trotter steps or higher-order product formulas~\cite{Suzuki1976Generalized, Hatano2005Finding}, thereby making practical implementation on quantum hardware increasingly challenging.

This difficulty has motivated a broad range of methods for reducing the implementation cost of Trotter-based simulation, including optimised product formulas~\cite{Somma2016A, Su2021Nearly, Stryker2025Shearing, Yuan2021Quantum}, randomised gate sequences~\cite{Childs2019Faster, Campbell2019Random, Yang2023Randomized, Kiumi2025TE-PAI}, per-step correction of Trotter errors~\cite{Tran2021Faster, Zeng2025Simple}, and quantum error mitigation (QEM) approaches~\cite{Endo2019Mitigating, Hakkaku2025Data-Efficient}.
While these approaches reduce the cost of implementing product formulas, they largely retain the conventional practice of decomposing the Hamiltonian based on commutativity or simple locality in a chosen representation that may not effectively reflect the physical structures governing the underlying dynamics, such as frustrated geometry or local symmetries.
This possible mismatch between the physical structure of the target dynamics and the algebraic structure of the chosen Hamiltonian representation has remained largely unexplored as a major source of residual Trotter error.

In this work, we address this underexplored source of overhead by introducing a symmetry-guided decomposition strategy beyond commutativity-based criteria for product-formula simulation.
Instead of partitioning the Hamiltonian according to the commutativity of individual terms, we group them into local three-site cluster Hamiltonians that respect the $\mathrm{SU}(2)$ symmetry of the underlying local dynamics.
This physically aligned decomposition substantially reduces the number of clusters in the partitioned Hamiltonian, thus directly suppressing residual inter-cluster errors.
At the same time, the inherent $\mathrm{SU}(2)$ symmetries in each triangular cluster enable efficient circuit implementations of the three-site propagators.
Consequently, the proposed product formula reduces both the number of sequential propagators per Trotter step and the circuit depth required to achieve a target accuracy.

For local triangular-plaquette propagators, we develop a systematic circuit construction for three-qubit $\mathrm{SU}(8)$ operators based on their $\mathrm{SU}(2)$ structure.
We show that the possible $\mathrm{SU}(2)$-symmetry structures are exhausted by at most four distinct classes, each allowing the relevant three-qubit Hamiltonian to be compressed into an effective two-qubit Hamiltonian with constant-overhead time evolution via two-qubit KAK decomposition~\cite{KAK, KAK1, KAK2, KAK3}.
When the triangular-plaquette Hamiltonian belongs to a single local $\mathrm{SU}(2)$ class, as in many frustrated spin models of interest, this construction bypasses the overhead of generic three-qubit KAK decompositions~\cite{Krol2024Beyond} or Schur-transform-based compression~\cite{Bacon2006Efficient, Nguyen2023The}.
Thus, our decomposition strategy is designed to expose local $\mathrm{SU}(2)$ structures that admit efficient triangular-plaquette implementations, with the flexibility to hybridise them with conventional commutativity-based groupings when advantageous.

We demonstrate the benefits of the proposed decomposition across several representative physical systems, deriving explicit decompositions and corresponding error analyses for each case.
Remarkably, we show that the proposed decomposition admits a two-cluster partitioning for which commutativity-based partitioning yields more clusters, enabling the second-order product formula to be implemented without additional circuit overhead.
Moreover, we further find that our decomposition based on local symmetries can also retain global symmetries that commutativity-based partitions may break at finite Trotter step size, as illustrated by spin-chirality dynamics.
These advantages are demonstrated numerically using a 12-qubit kagome Heisenberg model with three-body spin-chirality interactions, where the proposed decomposition reduces the errors in both state infidelity and average spin-chirality by more than three orders of magnitude compared with the conventional decomposition.
Our framework thus establishes local symmetry as a new design principle for product-formula decompositions, providing a practical route toward more accurate and circuit-efficient simulations of demanding physical models that are otherwise difficult to access with conventional approaches.

\section{\texorpdfstring{$\mathrm{SU}(2)$}{SU(2)}-classified triangular decomposition}

The product formula method approximates the time evolution $e^{-iHt}$ by partitioning the Hamiltonian into simpler, generally non-commuting components $H = H_{1}+H_{2}+\cdots+H_{m}$, and replacing the full short-time evolution by a sequential product of local evolutions.
For example, defining $\Delta t = t/n$, the first-order product formula is given by
\begin{equation}
\begin{split}
    e^{-iHt} 
    \approx \left(e^{-iH_{1}\Delta t}e^{-iH_{2}\Delta t}\cdots e^{-iH_{m}\Delta t}\right)^{n},
\end{split}
\end{equation}
which becomes exact in the limit $n\to\infty$ under the usual Trotter convergence conditions.

Without loss of generality, suppose that the Hamiltonian is expressed in the Pauli basis as $\displaystyle H = \sum_{P\in\mathcal{P}_{N}} c_{P}P$, where $c_{P}\in\mathbb{R}$ and $\mathcal{P}_{N}$ denotes the set of $N$-qubit Pauli strings.
The conventional commutativity-based Trotter decomposition partitions $H$ according to the commutation relations among the Pauli terms appearing in this chosen representation.
That is, mutually commuting terms are grouped into the same component so that each group can be implemented efficiently as a collection of Pauli rotations.
However, the effectiveness of such a decomposition can depend strongly on the chosen representation of $H$, and the resulting commuting groups need not reflect the physical structure governing the underlying dynamics.

\begin{figure}
    \centering
    \includegraphics[width=\linewidth]{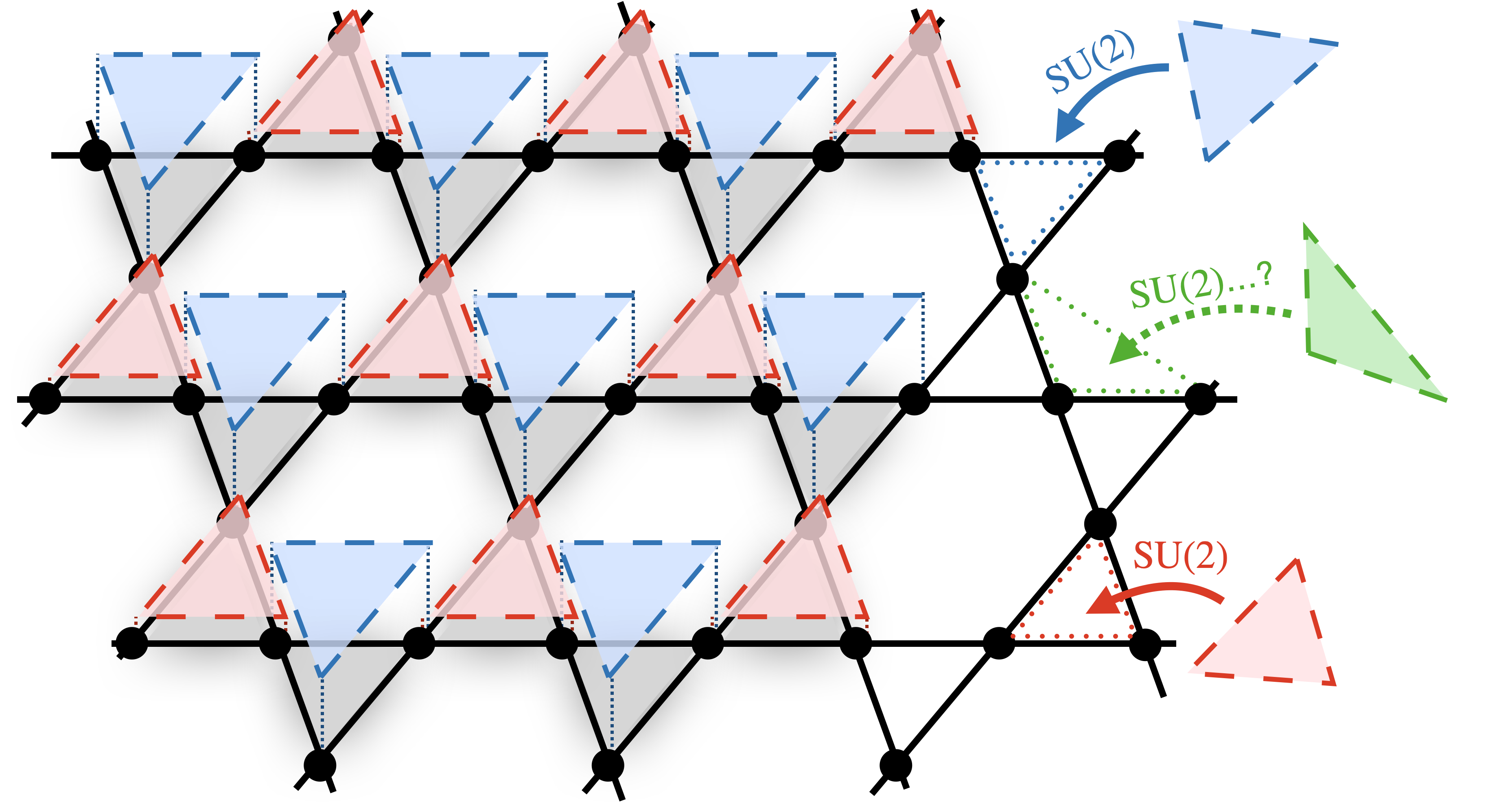}
    \caption{
        Schematic illustration of an $\mathrm{SU}(2)$-respecting triangular decomposition on a Kagome lattice.
        The Hamiltonian is partitioned into local three-site components supported on triangular plaquettes, shown as upward red and downward blue triangles, whose associated time-evolution operators preserve the underlying local $\mathrm{SU}(2)$ symmetry among the three vertices of each triangle structure.
    }
    \label{fig:plaquettes_triangles_kagome_2d}
\end{figure}

Our decomposition strategy goes beyond this commutativity-based criterion.
Rather than enforcing pairwise commutativity among the Pauli terms within each cluster, our method imposes a symmetry-based algebraic structure on the resulting cluster Hamiltonian.
Each cluster Hamiltonian is chosen either to commute with all generators in a certain local $\mathrm{SU}(2)$ generator space, or to admit a component-wise decomposition into sub-cluster Hamiltonians, each of which commutes with all generators in a possibly different local $\mathrm{SU}(2)$ generator space.
More formally, our decomposition strategy partitions the given Hamiltonian so that each clustered Hamiltonian $H_{\mathrm{C}}$ satisfies $[H_{\mathrm{C}}, G] = 0$ for all $G\in\mathcal{G}_{l}$, where $\mathcal{G}_{l}$ denotes the generator space of a possible local $\mathrm{SU}(2)$ operators characterised by an index $l$.

The decomposition procedure is therefore to identify a suitable partitioning of the target Hamiltonian into local components that preserve such $\mathrm{SU}(2)$ symmetries.
Note that this criterion is not tied to a particular basis or to a specific Pauli expansion of the Hamiltonian; it applies whenever a local three-site Hamiltonian admits the relevant $\mathrm{SU}(2)$-symmetry structure.
A schematic illustration of the proposed decomposition strategy is shown in Fig.~\ref{fig:plaquettes_triangles_kagome_2d}.

Our main contribution is to provide a concrete and efficient decomposition strategy for three-site local interactions by systematically classifying and exploiting the local $\mathrm{SU}(2)$ symmetries that arise in three-qubit $\mathrm{SU}(8)$ operations.
Below, we show that the possible local $\mathrm{SU}(2)$-symmetry structures of three-qubit $\mathrm{SU}(8)$ operations are exhausted by at most four distinct classes.
Equivalently, the three-qubit generator space can be spanned by at most four symmetry-resolved sectors, each decomposing into a local $\mathrm{SU}(2)$ generator space and its commuting effective $\mathrm{SU}(4)$ sector.
This classification is formalised in Theorem~\ref{theorem:classification_of_SU(8)_by_SU(2)}, and the resulting algebraic structure is illustrated in Fig.~\ref{fig:Upper_bound}.
The proof of Theorem~\ref{theorem:classification_of_SU(8)_by_SU(2)} is given in Appendix~\ref{appendix:Theorem_genflame}.

\begin{theorem}[Four-class decomposition by local $\mathrm{SU}(2)$ symmetry]
\label{theorem:classification_of_SU(8)_by_SU(2)}
Let $\mathcal{V}$ denote the real vector space of traceless Hermitian generators of three-qubit $\mathrm{SU}(8)$ evolutions.
There exist four local $\mathrm{SU}(2)$ symmetry classes such that
\begin{equation}
\begin{split}
    \mathcal{V}
    = \sum_{l=1}^{4} \left(\mathcal{G}_{l}\oplus\mathcal{H}_{l}\right),
\end{split}
\end{equation}
where $\mathcal{G}_{l}$ is the generator space of the $l$-th local $\mathrm{SU}(2)$ symmetry class, and $\mathcal{H}_{l}$ is the corresponding generator space of effective $\mathrm{SU}(4)$ operations commuting with $\mathcal{G}_{l}$.
Consequently, any three-qubit local Hamiltonian $H\in\mathcal{V}$ can be decomposed as
\begin{equation}\label{eq:decomp_principle}
\begin{split}
    H = \sum_{l=1}^{4} H_{l},
    \qquad
    H_{l}\in\mathcal{G}_{l}\oplus\mathcal{H}_{l},
\end{split}
\end{equation}
\end{theorem}

\begin{figure}
    \centering
    \includegraphics[width=0.5\linewidth]{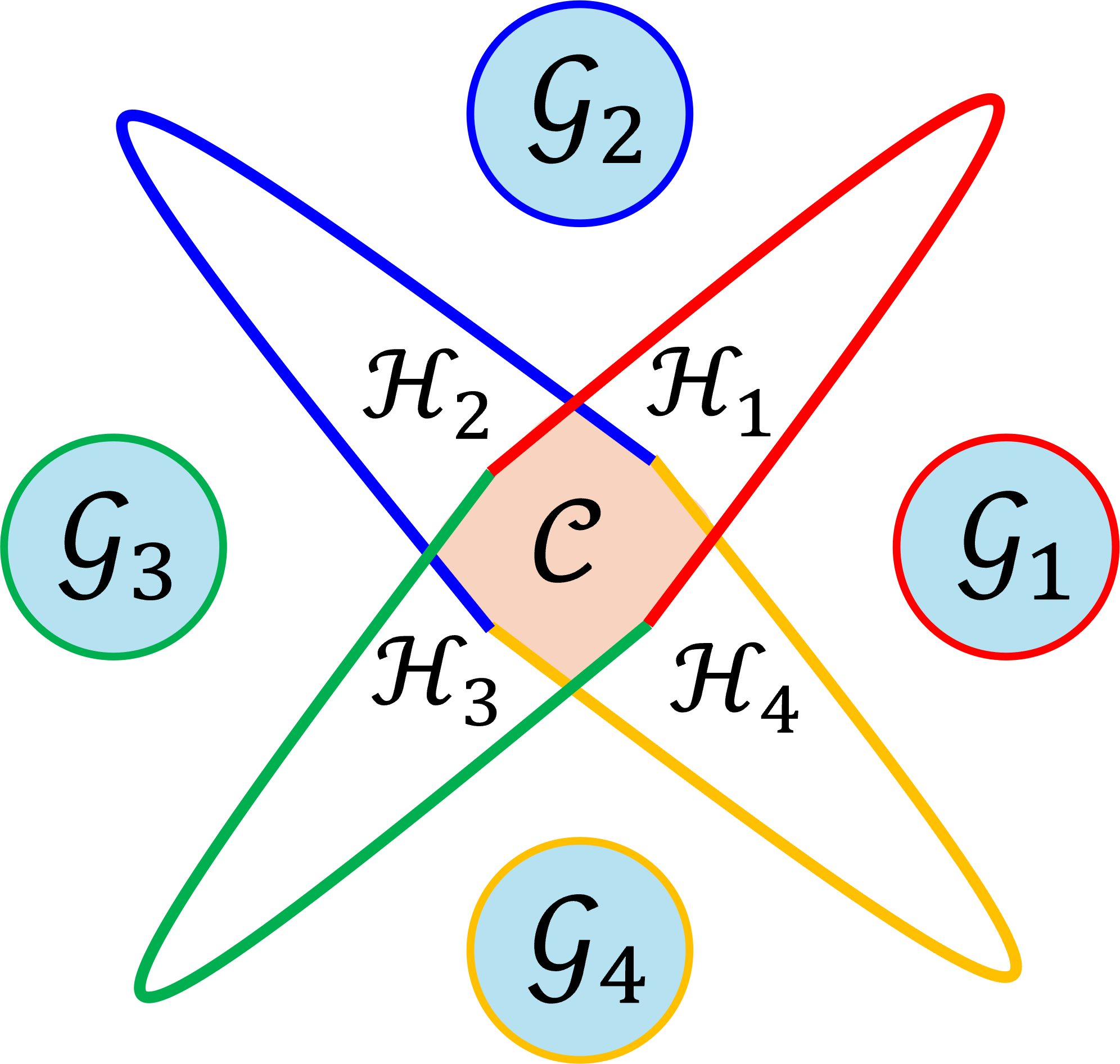}
    \caption{
        Algebraic structure of the generator space of three-qubit $\mathrm{SU}(8)$ operators.
        The orange region $\mathcal{C}$ denotes the common three-dimensional Abelian generator space shared by the effective two-qubit sectors $\mathcal{H}_{l}$.
        Each light-blue circular region $\mathcal{G}_{l}$ denotes the generator space of a local $\mathrm{SU}(2)$ symmetry class.
    }
    \label{fig:Upper_bound}
\end{figure}

This implies that, if the operator space generated by $\mathcal{G}_{l}$ and $\mathcal{H}_{l}$ can be factorised, by a unitary encoding, into a tensor-product structure $\mathrm{SU}(2)\otimes\mathrm{SU}(4)$, the essential part of the local three-qubit time evolution can be confined to an effective two-qubit subsystem.
More precisely, the $\mathrm{SU}(2)$ sector contributes only a single-qubit symmetry rotation, while the non-trivial dynamics is carried by the effective $\mathrm{SU}(4)$ evolution.

To realise this reduction, we next construct a unitary encoder that maps the effective $\mathrm{SU}(4)$ operator into a local two-qubit subspace.
Let $U_{l}\in \mathrm{SU}(8)$ be the unitary operator for the encoder of the symmetry class $l$.
Let also $G_{l}^{\prime}$ and $H_{l}^{\prime}$ be the single-qubit and two-qubit local effective Hamiltonians transformed by the unitary encoder $U_{l}$, i.e., $U_{l}H_{l}U_{l}^{\dagger} = G_{l}^{\prime}\otimes I\otimes I + I \otimes H_{l}^{\prime}$. 
Note that $U_{l}$ transforms the propagator of $H_{l}$ in the form of
\begin{equation}\label{eq:U_enc_U_dec}
\begin{split}
    U_{l} e^{-iH_{l}\Delta t} U_{l}^{\dagger}
    = e^{-iG_{l}^{\prime}\Delta t} \otimes e^{-iH_{l}^{\prime}\Delta t}.
\end{split}
\end{equation}
The quantum circuit of Eq.~\eqref{eq:U_enc_U_dec} is depicted in Fig.~\ref{fig:qc_trotter_block_triangle}.

\begin{figure}
    \centering
    \includegraphics[width=1.0\linewidth]{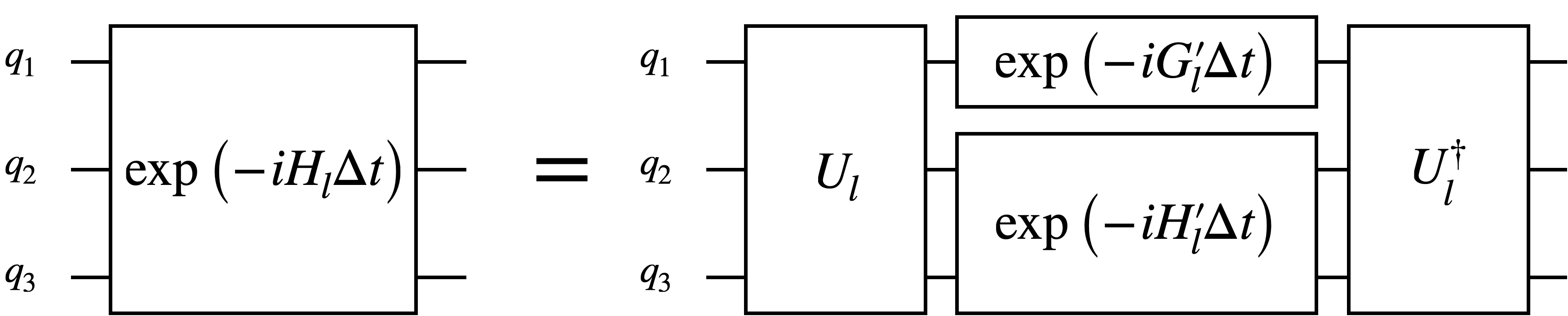}
    \caption{
        Quantum circuit of the local propagator $e^{-iH_{l}\Delta t}$ implemented by Eq.~\eqref{eq:U_enc_U_dec}.
        The encoded evolution by $U_{l}$ factorises into a single-qubit evolution $\exp(-iG_{l}^{\prime}\Delta t)$ and an effective two-qubit evolution $\exp(-iH_{l}^{\prime}\Delta t)$.
    }
    \label{fig:qc_trotter_block_triangle}
\end{figure}

Since we work in the Pauli representation, we may choose representative generator spaces $\{\mathcal{G}_{l}\}_{l=1}^{4}$ in the Pauli basis as
\begin{equation}\label{eq:G1_G2_G3_G4}
\begin{split}
    \mathcal{G}_{1} &= \operatorname{span}_{\mathbb{R}}\{X_{1} I_{2} I_{3},\, Y_{1} I_{2} I_{3},\, Z_{1} I_{2} I_{3}\},\\
    \mathcal{G}_{2} &= \operatorname{span}_{\mathbb{R}}\{X_{1} X_{2} X_{3},\, Y_{1} Y_{2} Y_{3},\, Z_{1} Z_{2} Z_{3}\}, \\
    \mathcal{G}_{3} &= \operatorname{span}_{\mathbb{R}}\{X_{1} Y_{2} Y_{3},\, Y_{1} Z_{2} Z_{3},\, Z_{1} X_{2} X_{3}\}, \\    
    \mathcal{G}_{4} &= \operatorname{span}_{\mathbb{R}}\{X_{1} Z_{2} Z_{3},\, Y_{1} X_{2} X_{3},\, Z_{1} Y_{2} Y_{3}\}.
\end{split}
\end{equation}
Here,
\begin{equation}
\begin{split}
    \operatorname{span}_{\mathbb{R}}\{e_{1},e_{2},\cdots,e_{k}\}
    := \left\{\sum_{i=1}^{k}a_{i}e_{i}~\middle|~a_{i}\in\mathbb{R}\right\}.
\end{split}
\end{equation}

For each symmetry class, we can choose the encoder $U_{l}$ to map the generator space $\mathcal{G}_{l}$ to the standard single-qubit generator space $\mathcal{G}_{1} = \operatorname{span}_{\mathbb{R}}\{X_{1}I_{2}I_{3},\,Y_{1}I_{2}I_{3},\,Z_{1}I_{2}I_{3}\}$, so that $G_{l}^{\prime}$ becomes trivial.
Equivalently, for every generator $G\in\mathcal{G}_{l}$, $U_{l}GU_{l}^{\dagger}\in \mathcal{G}_{1}$.
This mapping isolates the local $\mathrm{SU}(2)$ symmetry sector on a single qubit and embeds the commuting effective $\mathrm{SU}(4)$ sector into the remaining two-qubit subsystem.

\begin{figure}
    \centering
    \includegraphics[width=1.0\linewidth]{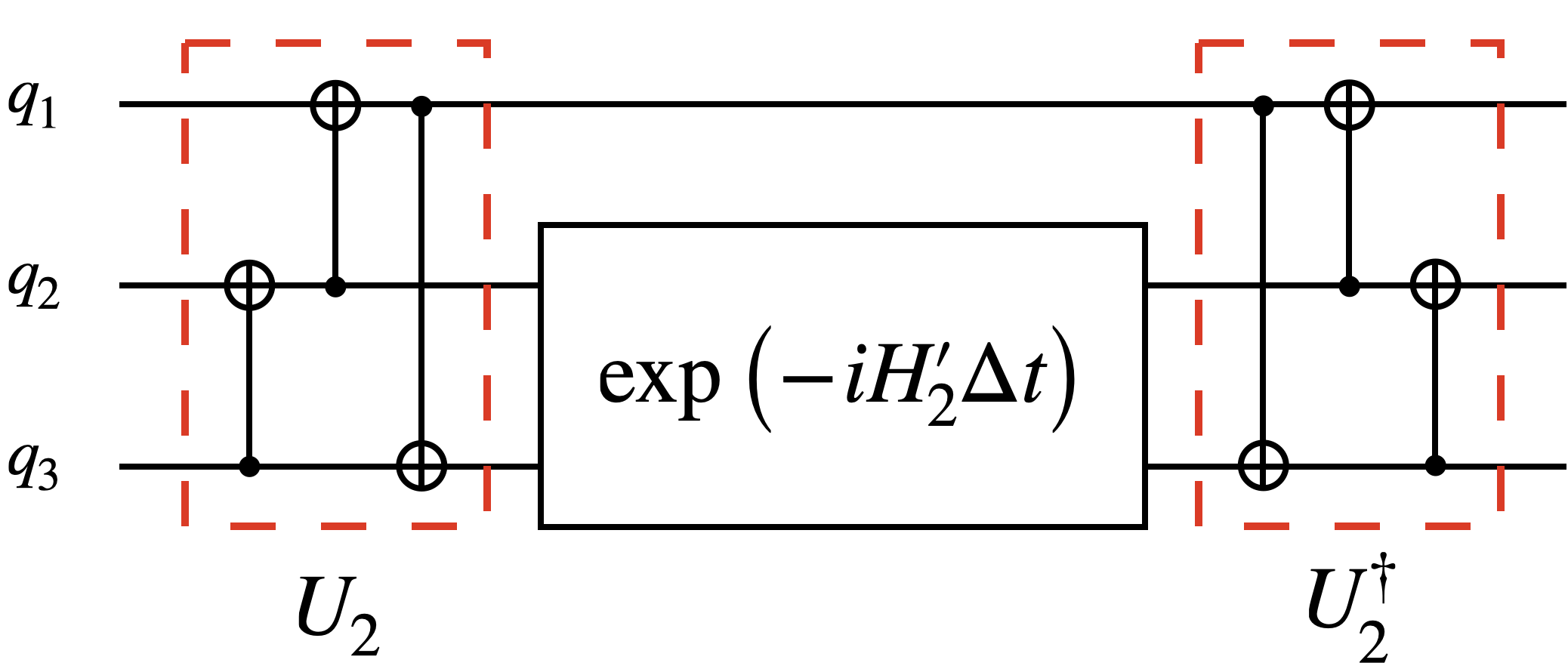}
    \caption{
        Circuit implementation of the local propagator $\exp(-iH_{2}\Delta t)$ for the local $\mathrm{SU}(2)$ symmetry class $l=2$.
        The red dashed block denotes the encoder $U_{2}$ by three CNOT gates.
    }
    \label{fig:qc_trotter_block_triangle_U2}
\end{figure}

As for the specific construction of $U_{l}$, for $l=1$, the encoding is trivial, and we take $U_{1}=I$.
For $l=2$, the encoder $U_{2}$ mapping $\mathcal{G}_{2}$ to $\mathcal{G}_{1}$ can be realised by the three CNOT gates shown in Fig.~\ref{fig:qc_trotter_block_triangle_U2}.
For $l=3$, one can reduce $\mathcal{G}_{3}$ to $\mathcal{G}_{2}$ by applying a local Clifford operation $\sqrt{Z_{1}}\sqrt{Y_{1}}$ on the first qubit, where $(\sqrt{Z_{1}})^{2} = Z_{1}$ and $(\sqrt{Y_{1}})^{2} = Y_{1}$.
Thus, the unitary $U_{3}$ that maps $\mathcal{G}_{3}$ to $\mathcal{G}_{1}$ is realised by
\begin{equation}
\begin{split}
    U_{3}=U_{2}\sqrt{Z_{1}}\sqrt{Y_{1}}.
\end{split}
\end{equation}
Similarly, for $l=4$, we obtain the encoding unitary
\begin{equation}
\begin{split}
    U_{4}=U_{2}\sqrt{Y_{1}}\sqrt{Z_{1}}
\end{split}
\end{equation}
to map $\mathcal{G}_{4}$ to $\mathcal{G}_{1}$.
These encoders provide one explicit construction of the unitary maps required to realise the $\mathrm{SU}(2)\otimes\mathrm{SU}(4)$ factorisation of each symmetry class.
Using these encoding unitary operators $U_{l}$, the propagator $\exp(-iH\Delta t)$ for a single time step $\Delta t$ is decomposed into
\begin{equation}\label{eq:decomp_gen}
\begin{split}
    & \exp(-iH\Delta t) \\
    &= \prod_{l=1}^{4}
       \exp\left(-iH_{l}\Delta t\right)
       + O((\Delta t)^{2}) \\
    &= \prod_{l=1}^{4}
       U_{l}^{\dagger}
       \left(\exp\left(-iG_{l}^{\prime}\Delta t\right)\otimes\exp\left(-iH_{l}^{\prime}\Delta t\right)\right)
       U_{l} + O((\Delta t)^{2}).
\end{split}
\end{equation}
Note that the decomposition in Eq.~\eqref{eq:decomp_gen} may involve additional Trotter errors among different symmetry classes.

While Theorem~\ref{theorem:classification_of_SU(8)_by_SU(2)} gives a generic classification of symmetry classes that is applicable to any three-qubit local Hamiltonian, the generator space of local Hamiltonians can generally be exhausted by less than four classes.
We further characterise the algebraic structure of such a case in Appendix~\ref{appendix:Theorem_complete}.

More importantly, when the three-qubit Hamiltonian belongs to a single local $\mathrm{SU}(2)$ symmetry class, the local propagator can be implemented exactly by only one $\mathrm{SU}(2)\otimes\mathrm{SU}(4)$ block as shown in Fig.~\ref{fig:qc_trotter_block_triangle_U2}, without the inter-class Trotter error appearing in Eq.~\eqref{eq:decomp_gen}.
We also note that the quantum circuit of each encoded propagator is significantly more efficient than the three-qubit KAK decomposition~\cite{Krol2024Beyond} that consumes $19$ CNOT gates, or using the Schur transform for the encoder~\cite{Bacon2006Efficient, Nguyen2023The}.

\section{Advantages of \texorpdfstring{$\mathrm{SU}(2)$}{SU(2)}-classified decomposition}

Through several representative many-body systems, we demonstrate that partitioning Hamiltonians according to their underlying local $\mathrm{SU}(2)$ symmetry structure can substantially reduce both circuit overhead and approximation error compared with conventional decompositions.
This improvement is especially pronounced when the decomposition is chosen so that each local Hamiltonian is confined to the generator space associated with a single $\mathrm{SU}(2)$ symmetry class.
We further show that the proposed decomposition naturally preserves certain global conservation laws that are broken by conventional approaches.

\subsection{Exact, efficient propagation within local three-site systems}

The proposed method achieves exact time propagation with significantly small circuit overhead for three-qubit Hamiltonians satisfying the $\mathrm{SU}(2)$ symmetry.
Given a three-qubit system with vertices $i$, $j$, and $k$, we consider the two-body interactions by the $XXX$ Heisenberg Hamiltonian among each pair of vertices, the three-body spin-chirality interactions among the three vertices, and the combined interactions of these two cases.

\subsubsection{Two-body interactions by $XXX$ Heisenberg Hamiltonian}

First, we consider the triangular quantum spin Hamiltonian 
\begin{equation}
\begin{split}
    H_{\mathrm{Heis}}^{(3)} = \vec{\sigma}_{i} \cdot \vec{\sigma}_{j} + \vec{\sigma}_{j} \cdot \vec{\sigma}_{k} + \vec{\sigma}_{k} \cdot \vec{\sigma}_{i},
\end{split}
\end{equation}
where $\vec{\sigma}_{i} \cdot \vec{\sigma}_{j}$ represents the two-body $XXX$ Heisenberg interaction term,
\begin{equation}
\begin{split}
    \vec{\sigma}_{i} \cdot \vec{\sigma}_{j} = X_{i}X_{j} + Y_{i}Y_{j} + Z_{i}Z_{j}.
\end{split}
\end{equation}
Based on the commutativity relation among the terms in $H_{\mathrm{Heis}}$, the conventional decomposition approximates the propagator $\exp\bigl(-iH_{\mathrm{Heis}}^{(3)}\Delta t\bigr)$ by a sequence of edge-associated propagators,
\begin{equation}\label{eq:convdecomp_Hcluster}
    \begin{split}
        \exp\left(-iH_{\mathrm{Heis}}^{(3)}\Delta t\right)
        &= \exp\left(-i\vec{\sigma}_{k}\cdot\vec{\sigma}_{i}\Delta t\right) \\
        &\quad\quad \times \exp\left(-i\vec{\sigma}_{j}\cdot\vec{\sigma}_{k}\Delta t\right) \\
        &\quad\quad \times \exp\left(-i\vec{\sigma}_{i}\cdot\vec{\sigma}_{j}\Delta t\right) \\
        &\quad+ O(\epsilon_{\mathrm{Heis}}(\Delta t)^{2}),
    \end{split}
\end{equation}
where the coefficient $\epsilon_{\mathrm{Heis}}$ in the residual term in Eq.~\eqref{eq:convdecomp_Hcluster} is defined by
\begin{equation}\label{eq:epsilon_{1}}
\begin{split}
    \epsilon_{\mathrm{Heis}}
    &=X_{i}(Y_{j}Z_{k} - Z_{j}Y_{k}) \\
    &\quad~ + Y_{i}(Z_{j}X_{k} - X_{j}Z_{k}) \\
    &\quad~ + Z_{i}(X_{j}Y_{k} - Y_{j}X_{k}).
\end{split}
\end{equation}
Since each propagator $\exp\left(-i\vec{\sigma}_{i}\cdot\vec{\sigma}_{j}\Delta t\right)$ is a two-qubit unitary, one can implement the sequence of propagators in Eq.~\eqref{eq:convdecomp_Hcluster} by three layers of quantum circuit in Fig.~\ref{fig:qcs_kagome12}(a).

On the other hand, the proposed decomposition directly implements the propagator $\exp\bigl(-iH_{\mathrm{Heis}}^{(3)}\Delta t\bigr)$.
Our method focuses on the fact that, for all edges in the triangle, representatively $(i, j)$, it holds that $[\vec{\sigma}_{i}\cdot\vec{\sigma}_{j}, P]=0$ for all $P\in\{X_{i}X_{j}X_{k}, Y_{i}Y_{j}Y_{k}, Z_{i}Z_{j}Z_{k}\}$.
That is, $H_{\mathrm{Heis}}^{(3)}$ satisfies the $\mathrm{SU}(2)$-symmetry regarding $\mathcal{G}_{2}$ in Eq.~\eqref{eq:G1_G2_G3_G4}.
Therefore, using the encoder $U_{2}$ in Fig.~\ref{fig:qc_trotter_block_triangle_U2}, we can transform $H_{\mathrm{Heis}}^{(3)}$ into $H^{\prime}$ defined by
\begin{equation}
\begin{split}
    H^{\prime}
    &= (X_{j}+Z_{j}) I_{k} + I_{j} (X_{k}+Z_{k}) \\
    &\quad+ (X_{j}-Z_{j}) (X_{k}-Z_{k}) + Y_{j}Y_{k},
\end{split}
\end{equation}
whose non-trivial time evolution is thus confined to the two-qubit system among vertices $j$ and $k$.
This implies that one can design an exact propagator $\exp\bigl(-iH_{\mathrm{Heis}}^{(3)}\Delta t\bigr)$ by sandwiching $H^{\prime}$ with the encoder and decoder as depicted in Fig.~\ref{fig:qc_trotter_block_triangle_U2}.

\subsubsection{Three-body spin-chirality interactions}

Next, we consider the three-body spin-chirality interaction Hamiltonian $H_{\mathrm{Chiral}}^{(3)} = \vec{\sigma}_{i} \cdot (\vec{\sigma}_{j} \times \vec{\sigma}_{k})$, defined by
\begin{equation}
\begin{split}\label{eq:def_threebody_unit}
    H_{\mathrm{Chiral}}^{(3)} = \vec{\sigma}_{i} \cdot (\vec{\sigma}_{j} \times \vec{\sigma}_{k}) 
    &:= X_{i}(Y_{j}Z_{k} - Z_{j}Y_{k}) \\
    &\quad~ + Y_{i}(Z_{j}X_{k} - X_{j}Z_{k}) \\
    &\quad~ + Z_{i}(X_{j}Y_{k} - Y_{j}X_{k}).
\end{split}
\end{equation}
We here remark that $\epsilon_{\mathrm{Heis}} = H_{\mathrm{Chiral}}^{(3)}$, which is obvious by comparing Eq.~\eqref{eq:epsilon_{1}} and Eq.~\eqref{eq:def_threebody_unit}.
Using the notation $\chi_{ijk}:=X_{i}(Y_{j}Z_{k}-Z_{j}Y_{k})$, Eq.~\eqref{eq:def_threebody_unit} can be simplified into
\begin{equation}
\begin{split}
    H_{\mathrm{Chiral}}^{(3)} = \chi_{ijk} + \chi_{kij} + \chi_{jki}.
\end{split}
\end{equation}

Since the conventional decomposition groups the commuting terms into the same cluster, it approximates the propagator $\exp\bigl(-iH_{\mathrm{Chiral}}^{(3)}\Delta t\bigr)$ by 
\begin{equation}\label{eq:convdecomp_Hcluster2}
    \begin{split}
        &\exp\left(-iH_{\mathrm{Chiral}}^{(3)}\Delta t\right) \\
        &= \exp\left(-i\chi_{jki}\Delta t\right)
            \exp\left(-i \chi_{kij}\Delta t\right)
            \exp\left(-i\chi_{ijk}\Delta t\right) \\
        &\quad+ O(\epsilon_{\mathrm{Chiral}}(\Delta t)^{2}),
    \end{split}
\end{equation}
where the residual term $\epsilon_{\mathrm{Chiral}}$ is given by
\begin{equation}\label{eq:epsilon_2}
\begin{split}
    \epsilon_{\mathrm{Chiral}}
    &=\chi_{kij}-\chi_{jki}-\chi_{ijk}.
\end{split}
\end{equation}
Since each propagator $\exp\left(-i\chi_{ijk}\Delta t\right)$ can be realised by the quantum circuit in Fig.~\ref{fig:qcs_kagome12}(b), $\exp\left(-iH_{\mathrm{Chiral}}^{(3)}\Delta t\right)$ is implemented by three layers of this circuit gadget.

On the other hand, again, the proposed decomposition can efficiently implement the propagator $\exp\bigl(-iH_{\mathrm{Chiral}}^{(3)}\Delta t\bigr)$, since $H_{\mathrm{Chiral}}^{(3)}$ also satisfies $[H_{\mathrm{Chiral}}^{(3)}, P]=0$ for all $P\in\{X_{i}X_{j}X_{k}, Y_{i}Y_{j}Y_{k}, Z_{i}Z_{j}Z_{k}\}$.
Using the encoder $U_{2}$ in Fig.~\ref{fig:qc_trotter_block_triangle_U2}, we can transform $H_{\mathrm{Heis}}^{(3)}$ into the following $H^{\prime\prime}$ confined to the two-qubit system among vertices $j$ and $k$, defined by
\begin{equation}
\begin{split}
    H^{\prime\prime}
    &= -Y_{j} I_{k}+I_{j}Y_{k} \\
    &\quad +Y_{j} (X_{k}+Z_{k})-(X_{j}+Z_{j})Y_{k},
\end{split}
\end{equation}
Therefore, one can design an exact propagator $\exp\bigl(-iH_{\mathrm{Chiral}}^{(3)}\Delta t\bigr)$ by sandwiching $H^{\prime\prime}$ with the encoder and decoder as depicted in Fig.~\ref{fig:qc_trotter_block_triangle_U2}.

\subsubsection{Combinations of the two types of interaction}

Now, we consider the circuit realisation of the propagator of the combined Hamiltonian $H_{\mathrm{Heis}}^{(3)}+H_{\mathrm{Chiral}}^{(3)}$.
From the above, we have checked that both $H_{\mathrm{Heis}}^{(3)}$ and $H_{\mathrm{Chiral}}^{(3)}$ satisfy the same $\mathrm{SU}(2)$ symmetry.
This immediately implies that the proposed method can exactly implement the propagator $\exp\bigl(-i(H_{\mathrm{Heis}}^{(3)}+H_{\mathrm{Chiral}}^{(3)}) \Delta t\bigr)$ using the same circuit in Fig.~\ref{fig:qc_trotter_block_triangle_U2} by sandwiching $\exp\left(-i(H^{\prime}+H^{\prime\prime})\Delta t\right)$ with $U_{2}$ encoder and decoder.

Meanwhile, when it comes to the conventional decomposition, it requires a further deeper circuit to approximate $\exp\bigl(-i(H_{\mathrm{Heis}}^{(3)}+H_{\mathrm{Chiral}}^{(3)}) \Delta t\bigr)$ by concatenating the circuit for $\exp\bigl(-iH_{\mathrm{Heis}}^{(3)}\Delta t\bigr)$ and $\exp\bigl(-iH_{\mathrm{Chiral}}^{(3)} \Delta t\bigr)$.
The approximation error also sums up to $\epsilon_{\mathrm{Hies}} + \epsilon_{\mathrm{Chiral}}$.
Note that $[H_{\mathrm{Heis}}^{(3)},H_{\mathrm{Chiral}}^{(3)}]=0$ since 
\begin{equation}
\begin{split}
    H_{\mathrm{Heis}}^{(3)}
    &= \vec{\sigma}_{i}\cdot\vec{\sigma}_{j}+\vec{\sigma}_{j}\cdot\vec{\sigma}_{k}+\vec{\sigma}_{k}\cdot\vec{\sigma}_{i} \\
    &= \frac{1}{2}(\vec{\sigma}_{i}+\vec{\sigma}_{j}+\vec{\sigma}_{k})^{2}- \frac{9}{2}I,
\end{split}
\end{equation}
and thus no additional Trotter error would occur when decomposing the propagator as below
\begin{equation}\label{eq:split_H2H3}
\begin{split}
    & \exp\bigl(-i(H_{\mathrm{Heis}}^{(3)}+H_{\mathrm{Chiral}}^{(3)}) \Delta t\bigr) \\
    &= \exp\bigl(-iH_{\mathrm{Chiral}}^{(3)} \Delta t\bigr)
    \exp\bigl(-iH_{\mathrm{Heis}}^{(3)} \Delta t\bigr).
\end{split}
\end{equation}

\begin{figure}
    \centering
    \subfloat[]{
        \includegraphics[width=0.47\linewidth]{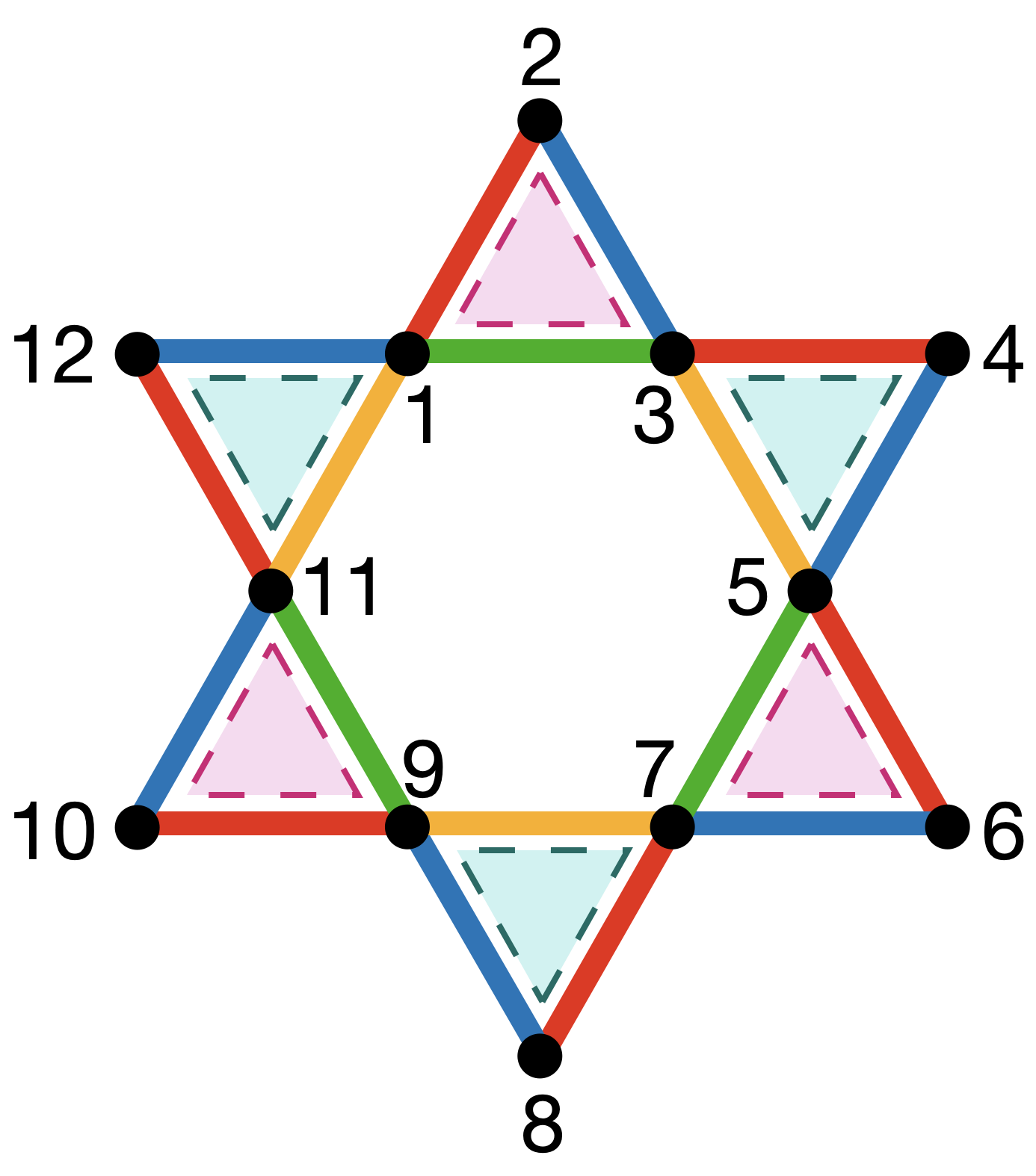}
    }
    \hfill
    \subfloat[]{
        \includegraphics[width=0.47\linewidth]{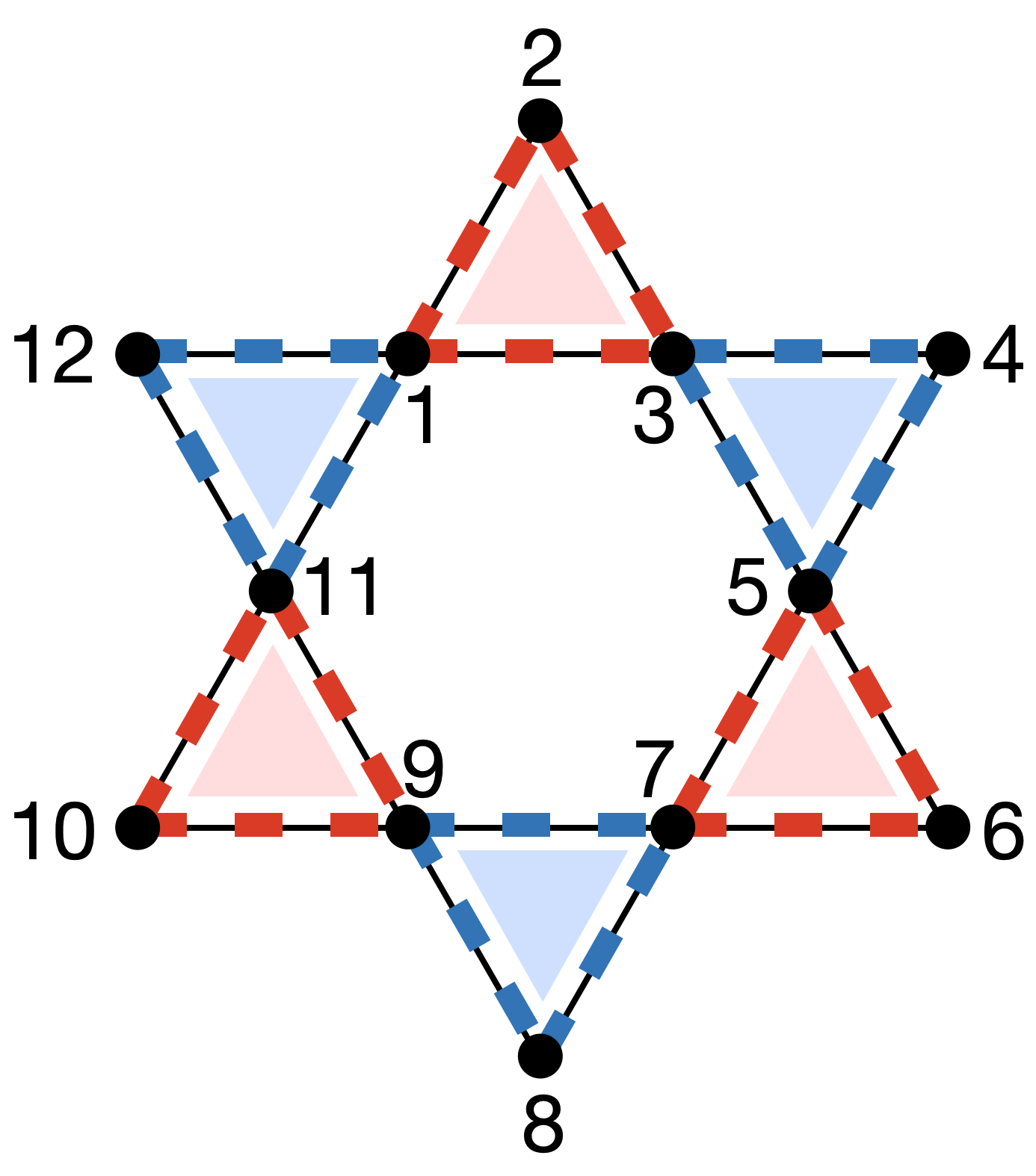}
    }
    \caption{
        The connectivity graph of the 12-qubit Kagome $XXX$ Heisenberg lattice with spin-chirality interactions within the triangular parts.
        (a) Clustering by the conventional decomposition.
        The two-qubit interactions $\vec{\sigma}_{i}\cdot\vec{\sigma}_{j}$ corresponding to the edges are partitioned into four clusters coloured in red, blue, green, and yellow.
        The spin-chirality interactions $\vec{\sigma}_{i}\cdot(\vec{\sigma}_{j}\times\vec{\sigma}_{k})$ within each triangle are partitioned into two clusters, coloured pink and light blue.
        In total, the conventional method partitions the interactions into the product formula with $4 + 3\times 2 = 10$ clusters.
        Each pink or light-blue part is further decomposed into three clusters.
        (b) Clustering by the proposed decomposition, which allows for the product formula into only two clusters, the red and the blue triangles.
    }
    \label{fig:decompositions_kagome_chirality_12}
\end{figure}

\begin{figure}
    \centering
    \subfloat[]{
        \includegraphics[width=\linewidth]{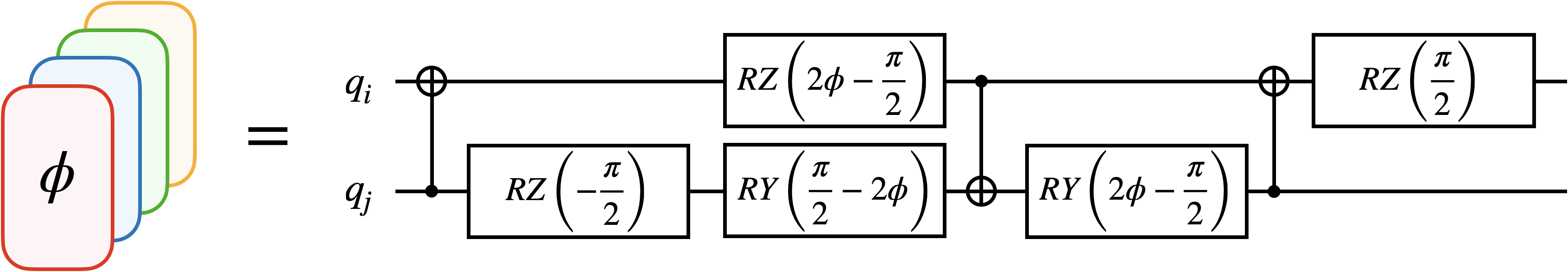}
    }
    \hfill
    \subfloat[]{
        \includegraphics[width=0.66\linewidth]{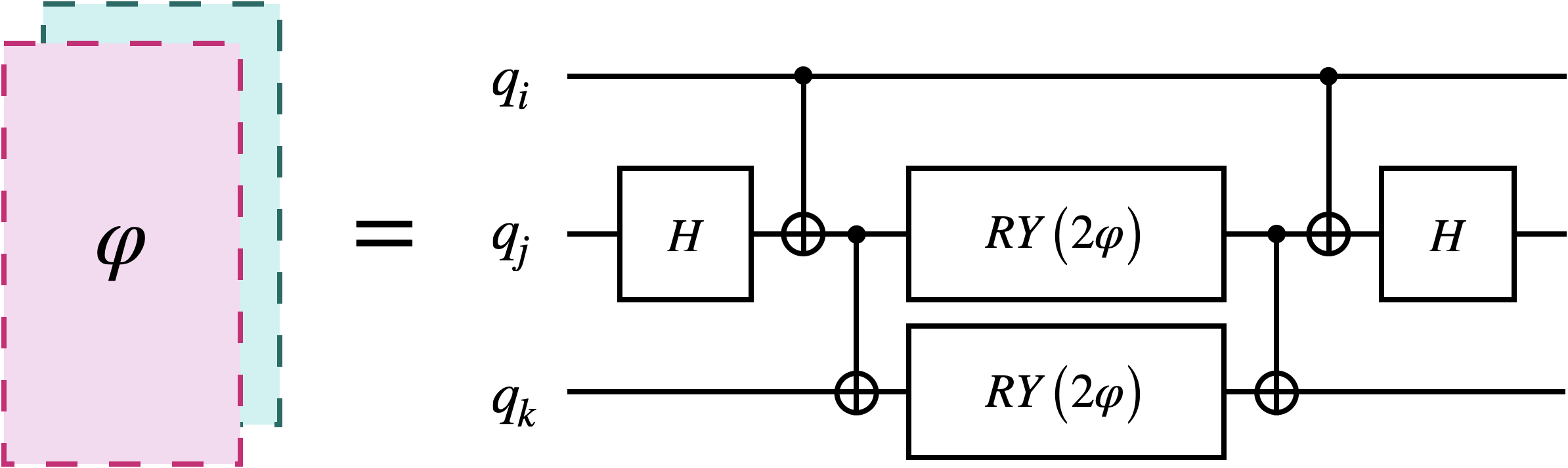}
    }
    \hfill
    \subfloat[]{
        \includegraphics[width=\linewidth]{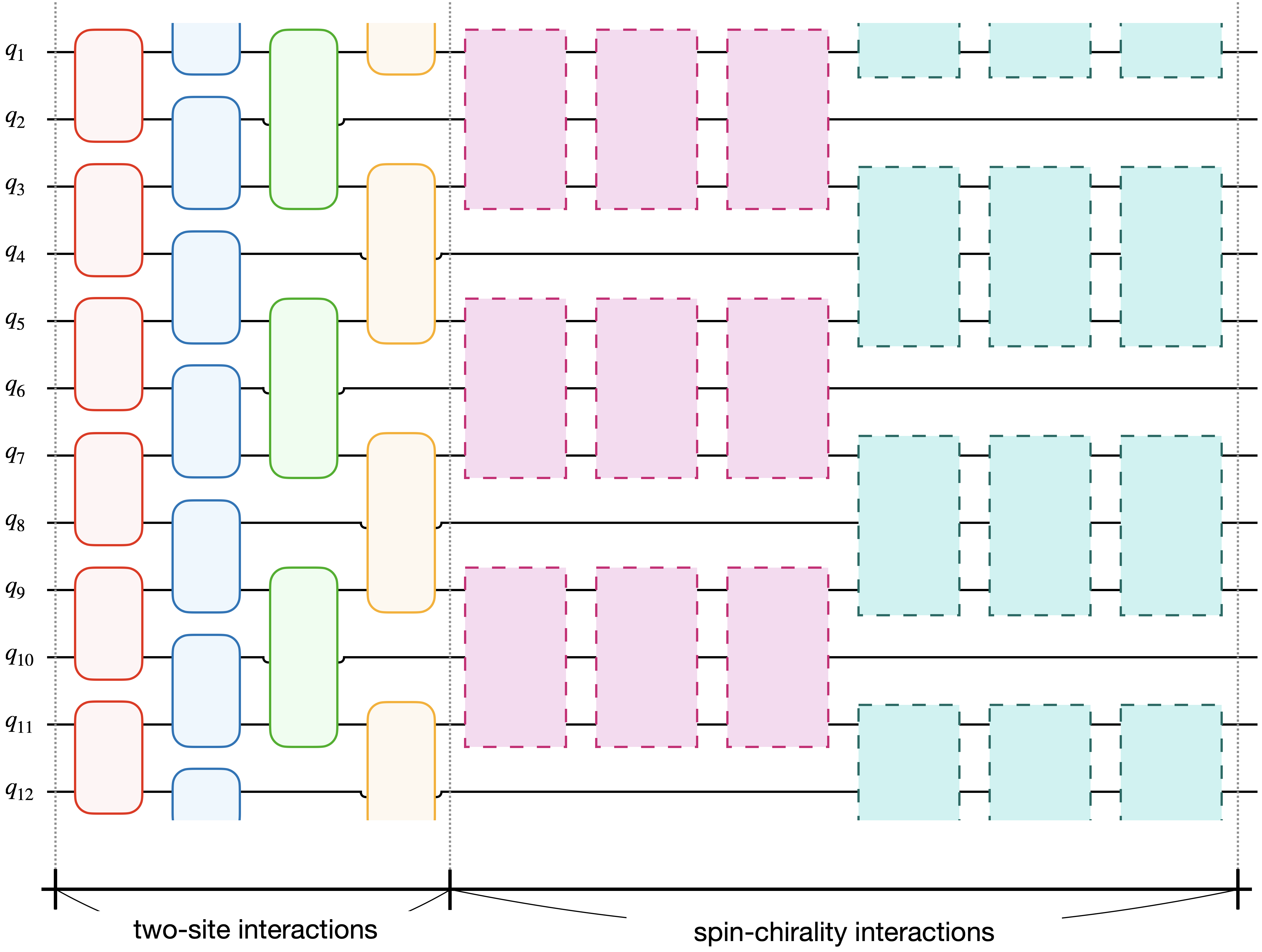}
    }
    \hfill
    \subfloat[]{
        \includegraphics[width=0.66\linewidth]{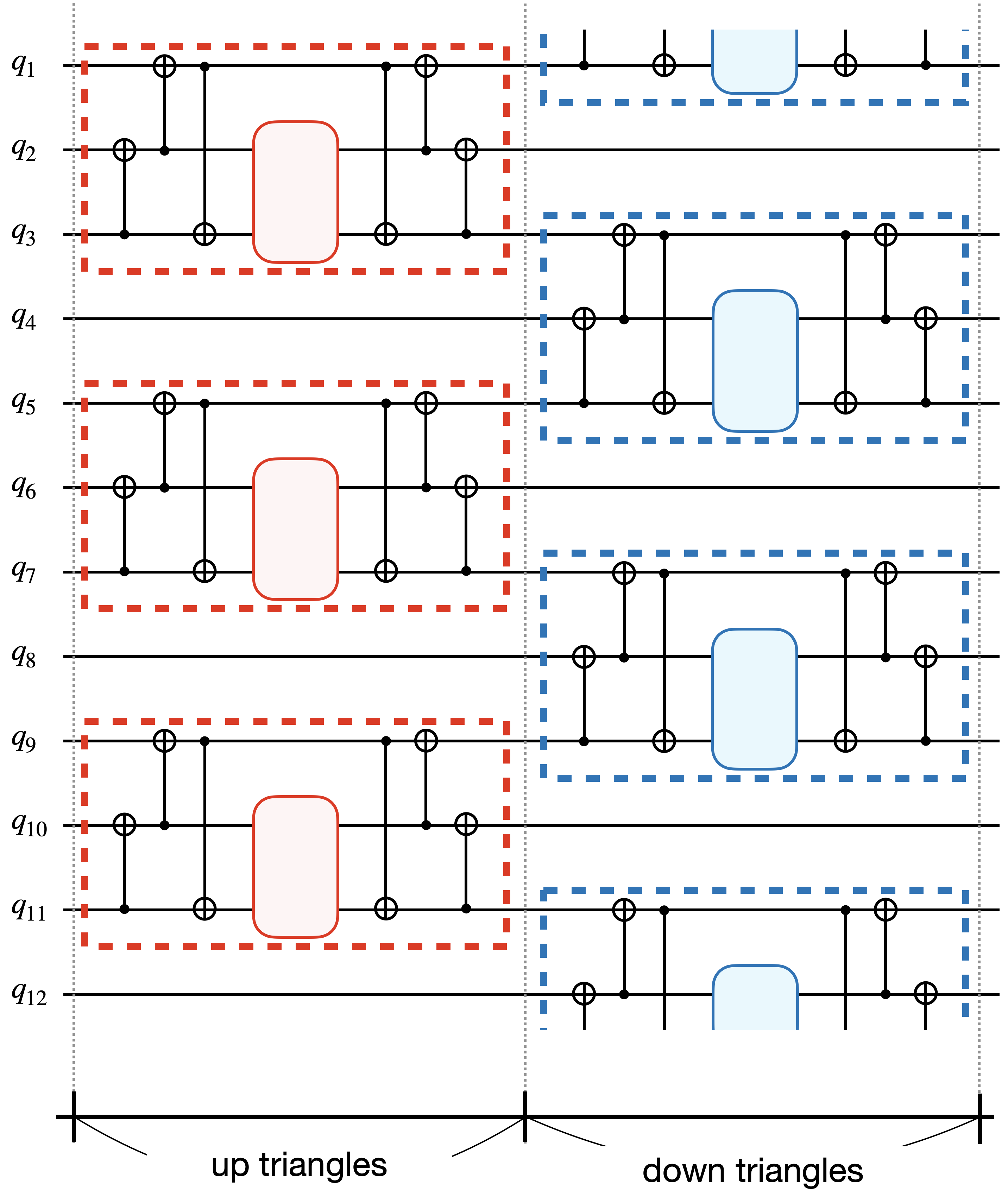}
    }
    \caption{
        Elementary circuits and first-order decomposition steps.
        (a) Circuit for the two-qubit exchange evolution $\exp\left(-i\phi\,\vec{\sigma}_{i}\cdot\vec{\sigma}_{j}\right)$.
        (b) Circuit for the three-qubit evolution $\exp\left(-i\varphi Z_{i}\left(X_{j}Y_{k}-Y_{j}X_{k}\right)\right)$.
        (c) A single first-order Trotter step obtained from the conventional commutativity-based decomposition.
        (d) A single first-order Trotter step obtained from the proposed $\mathrm{SU}(2)$-respecting triangular decomposition.
        In (c) and (d), different gate colours indicate the implementations of the corresponding clusters shown in the same colours in Fig.~\ref{fig:decompositions_kagome_chirality_12}.
    }
    \label{fig:qcs_kagome12}
\end{figure}

\subsection{Reducing inter-cluster errors and circuit overhead on the Kagome lattice}

Here, we further demonstrate that the proposed method substantially reduces both the inter-cluster residual errors and the total circuit overhead in a demanding many-body setting: the $XXX$ Heisenberg model on the Kagome lattice with triangular spin-chirality interactions.
The target Hamiltonian, denoted by $H_{\mathrm{Kagome}}$, is constructed by assigning the local interaction $H_{\mathrm{Heis}}^{(3)} + H_{\mathrm{Chiral}}^{(3)}$ to each triangular plaquette of the lattice.
The schematic illustration of a 12-site Kagome lattice is shown in Fig.~\ref{fig:decompositions_kagome_chirality_12}.

First, with the conventional decomposition, $H_{\mathrm{Kagome}}$ is decomposed into the $XXX$ Heisenberg interactions $H_{\mathrm{Heis}}$ and the spin-chirality interactions $H_{\mathrm{Chiral}}$, that is, $H_{\mathrm{Kagome}} = H_{\mathrm{Heis}} + H_{\mathrm{Chiral}}$.
To construct tractable propagators, $H_{\mathrm{Heis}}$ is further decomposed into four clusters corresponding to different matchings, i.e., a set
of pairwise vertex-disjoint edges.
Meanwhile, since the two different $H_{\mathrm{Chiral}}^{(3)}$ sharing only a single vertex do not commute with each other, $H_{\mathrm{Chiral}}$ is also decomposed into two clusters so that each cluster does not share common vertices.
This clustering is illustrated in Fig.~\ref{fig:decompositions_kagome_chirality_12}(a).
As a result, a single Trotter step by the conventional method requires ten Trotter blocks, which yields a deep quantum circuit as shown in Fig.~\ref{fig:qcs_kagome12}(c).

\begin{table*}[htbp]
    \centering
    \renewcommand{\arraystretch}{1.5}
    \begin{tabular}{ l|c|c|c|c }
        \hline\hline
        system Hamiltonian & \makecell{number \\of clusters \\ (conventional)} & \makecell{number \\of clusters \\ \textbf{(proposed)}} & \makecell{residual error\\ (conventional)} & \makecell{residual error\\ \textbf{(proposed)}} \\
        \hline
        1D transverse-field Ising chain& 2 & 2 & 2 & 0.5 \\
        1D Heisenberg chain~\cite{Yang2026Quantum} & 2 & 2 & 6 & 3 \\
        1D two-layer $J_{1}$--$J_{2}$ Heisenberg & 4 & 2 & $6+12t+6t^{2}$ & 3+12$t$+12$t^{2}$ \\
        Square lattice transverse-field Ising & 3 & 2 & 4 & 2 \\
        Square lattice Heisenberg & 4 & 2 & 36 & 30 \\
        Kagome lattice Heisenberg & 4 & 2& 28 & 24 \\
        \makecell[l]{Kagome lattice Heisenberg with triangular spin-chirality} & 10 & 2 &28+48$t+28t^{2}$&24+48$t$+24$t^{2}$ \\
        Triangular lattice Heisenberg & 6 & 3 & 66 & 60 \\
        \makecell[l]{Triangular lattice Heisenberg with triangular spin-chirality} & 24 & 3 & 66+288$t$+264$t^{2}$&60+288$t$+504$t^{2}$ \\
        \hline\hline
    \end{tabular}
    \caption{
        The list of representative Hamiltonian instances that benefit from our method.
        The columns for the ``residual error'' count the number of Pauli terms per vertex in the residual error.
        All these models admit the decomposition with $\mathcal{G}_{2}$ symmetry in Theorem~\ref{theorem:classification_of_SU(8)_by_SU(2)}.
        The decomposition strategies by the conventional method and the proposed method for each model are detailed in Appendix~\ref{appendix:examples}. The real value $t$ denotes the ratio of coupling constants defined as $t=J_{2}/J_{1}$ for a 1D two-layer Heisenberg chain and $t=K/J$ for a 2D lattice with triangular spin-chirality interaction.
    }
    \label{tab:list_of_Hamiltonians}
\end{table*}

By contrast, our approach decomposes $H_{\mathrm{Kagome}}$ into only two clusters,
\begin{equation}
\begin{split}
    H_{\mathrm{Kagome}} = H_{\bigtriangleup} + H_{\bigtriangledown},
\end{split}
\end{equation}
where $H_{\bigtriangleup}$ and $H_{\bigtriangledown}$ consist of the interactions within the red and blue triangles in Fig.~\ref{fig:decompositions_kagome_chirality_12}(b), respectively.
Each local red or blue triangle has the Hamiltonian $H_{\mathrm{Heis}}^{(3)} + H_{\mathrm{Chiral}}^{(3)}$, whose propagation can be exactly processed by the proposed method, as we have seen above.
Consequently, the resulting quantum circuit takes the much simpler form shown in Fig.~\ref{fig:qcs_kagome12}(d).

\subsection{Applicability in a broad class of lattice models}

Our approach is applicable not only to the Kagome lattice shown above, but also to a variety of quantum spin lattice models.
We list several representative lattice models in Table~\ref{tab:list_of_Hamiltonians}, comparing the number of clusters and the residual error, which counts the number of Pauli terms per vertex.
The detailed decomposition and derivation of the table elements are provided in Appendix~\ref{appendix:examples}.

From Table~\ref{tab:list_of_Hamiltonians}, we observe that the proposed method yields clearly fewer clusters, which directly contributes to reducing the circuit depth in each Trotter step to sequentially implement Trotter blocks corresponding to each cluster.
Remarkably, whenever the proposed method yields only two clusters, it allows implementing the second-order product formula without additional circuit cost, or at most a constant cost: it does not double the total circuit depth to suppress the residual error in magnitude by $\Delta t$.
Table~\ref{tab:list_of_Hamiltonians} implies that for many lattice models, the proposed method achieves the two-cluster decomposition, while the conventional method typically requires more than three.
The circuit construction for a higher-order product formula would also benefit from having fewer clusters.

Our method also exhibits a smaller residual error in the number of Puali terms averaged per vertex.
Note that this provides only a loose upper bound on the residual error, and a tighter comparison remains open.
We remark that the proposed decomposition can yield a much stronger suppression of residual errors than is apparent from this coarse counting estimate.
We provide numerical evidence for the Kagome lattice in Section~\ref{sec:numerics}, and discuss the physical implications of our design principle below.

\begin{figure*}[htbp]
    \centering
    \subfloat[]{
        \includegraphics[width=0.48\linewidth]{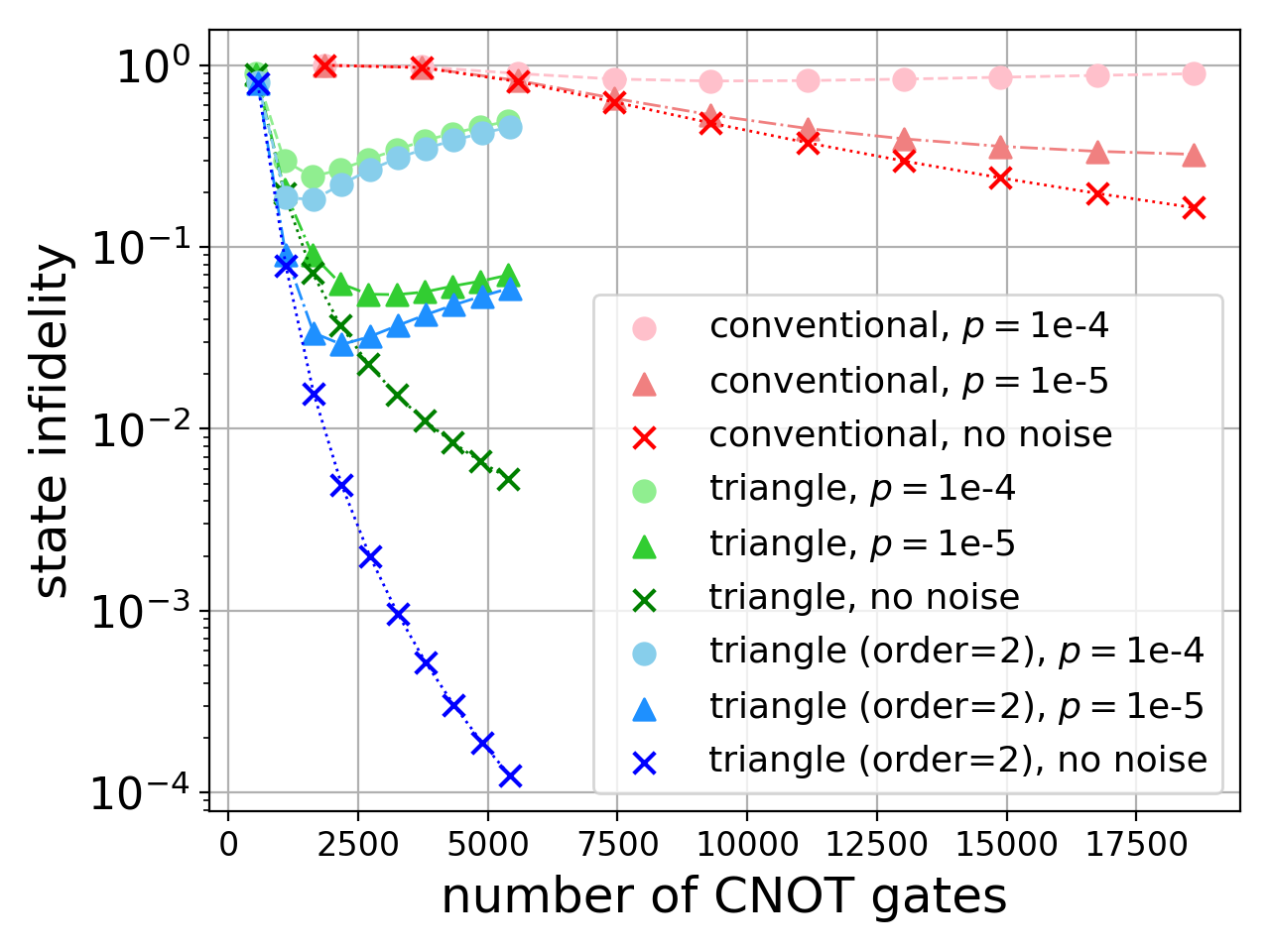}
    }
    \hfill
    \subfloat[]{
        \includegraphics[width=0.48\linewidth]{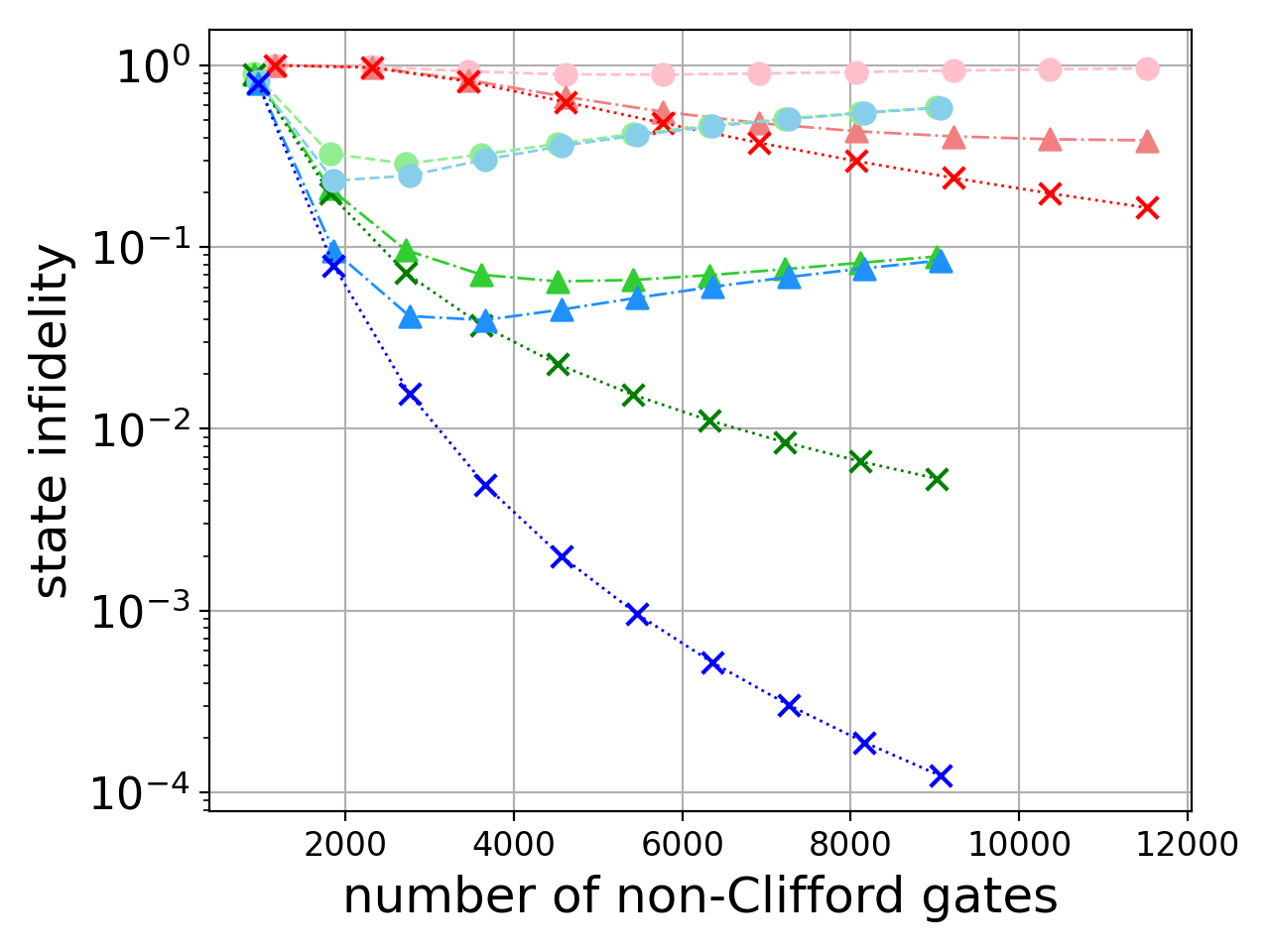}
    }
    \caption{
        State infidelity at the evolution time $t=\pi$ for the state evolved from $t=0$.
        The Trotter steps are investigated from $\{10, 20, \ldots, 100\}$.
        The vertical axis shows the infidelity, while the horizontal axis shows
        (a) the number of CNOT gates for each Trotter step,
        (b) the number of non-Clifford gates for each Trotter step.
        In figure (a), the results for the noise-free case, depolarising error with $(p_{1},p_{2})=(1.0\times10^{-6},\,1.0\times10^{-5})$, and depolarising error with $(p_{1},p_{2})=(1.0\times10^{-5},\,1.0\times10^{-4})$ are indicated by $\times$, $\blacktriangle$, and $\bullet$, respectively.
        In figure (b), the results for the noise-free case, dephasing error with $p_{z} = 1.0\times10^{-5}$, and dephasing error with $p_{z} = 1.0\times10^{-4}$ are indicated by $\times$, $\blacktriangle$, and $\bullet$, respectively.
        The conventional method, and the first- and second-order decompositions of the proposed method, are shown in shades of red, green, and blue, respectively.
    }
    \label{fig:kagome_ring_infidelity}
\end{figure*}

\subsection{Preserving global conservation laws via local symmetries}

By respecting the local $\mathrm{SU}(2)$symmetry, the proposed method provides not only operational advantages but also physically faithful time propagation.
In particular, the proposed method preserves the total spin angular momentum under time evolution governed by the $XXX$ Heisenberg and spin-chirality interactions, whereas the conventional approach violates it.

We illustrate this using a physical model whose Hamiltonian is spanned by the spin-chirality interaction, such as $\epsilon_{\mathrm{Chiral}}$ in Eq.~\eqref{eq:epsilon_2}.
First, we define the total spin operators $(S_{\mathrm{x}},S_{\mathrm{y}},S_{\mathrm{z}})$ for the $N$-site spin system by
\begin{equation}
\begin{split}
    (S_{\mathrm{x}},S_{\mathrm{y}},S_{\mathrm{z}})=\sum_{i=1}^{N}(X_{i},Y_{i},Z_{i}).
\end{split}
\end{equation}
Assuming that the Hamiltonian $H$ is generated by spin-chirality interactions, the total spin operators commute with $H$, which follows directly from the local commutation
relations on each triangle $(i,j,k)$:
\begin{equation}\label{eq:commute}
\begin{split}
    &[X_{i}+X_{j}+X_{k},\vec{\sigma}_{i} \cdot (\vec{\sigma}_{j} \times \vec{\sigma}_{k})]=0,\\
    &[Z_{i}+Z_{j}+Z_{k},\vec{\sigma}_{i} \cdot (\vec{\sigma}_{j} \times \vec{\sigma}_{k})]=0.
\end{split}
\end{equation}
Therefore, the total spin operators are conserved under time evolution since the commutation relation $[S_{i}, H]=0$ implies $\exp(iHt)S_{i}\exp(-iHt)=S_{i}$.

Our $\mathrm{SU}(2)$-classification strategy guarantees this conservation law because it exactly implements the propagator $\exp\bigl(-i\vec{\sigma}_{i} \cdot (\vec{\sigma}_{j} \times \vec{\sigma}_{k})\Delta t\bigr)$
for every triangle $(i,j,k)$. By contrast, the conventional decomposition produces residual terms given by Eq.~\eqref{eq:epsilon_2}, which do not commute with the total spin operators, i.e., $[\epsilon_{\rm Chiral}, S_{i}]\neq0$.
This non-commutativity directly leads to violations of the conservation laws during the Trotterised time evolution.

As we have seen above, the proposed method preserves the total spin angular momenta, which is well known to be conserved under the actual time evolution governed by the $XXX$ Heisenberg and spin-chirality interactions.
Since such violations can significantly degrade the accuracy of the dynamics simulation, existing works address this issue by inserting additional gadgets to correct per-step residual errors, with a focus on global symmetry~\cite{Tran2021Faster, Zeng2025Simple}.
While such approaches generally require additional circuit overhead, our design principle based on local $\mathrm{SU}(2)$ symmetry naturally preserves physically important conservation laws while simultaneously reducing computational overhead.

\begin{figure*}[htbp]
    \centering
    \subfloat[]{
        \includegraphics[width=0.48\linewidth]{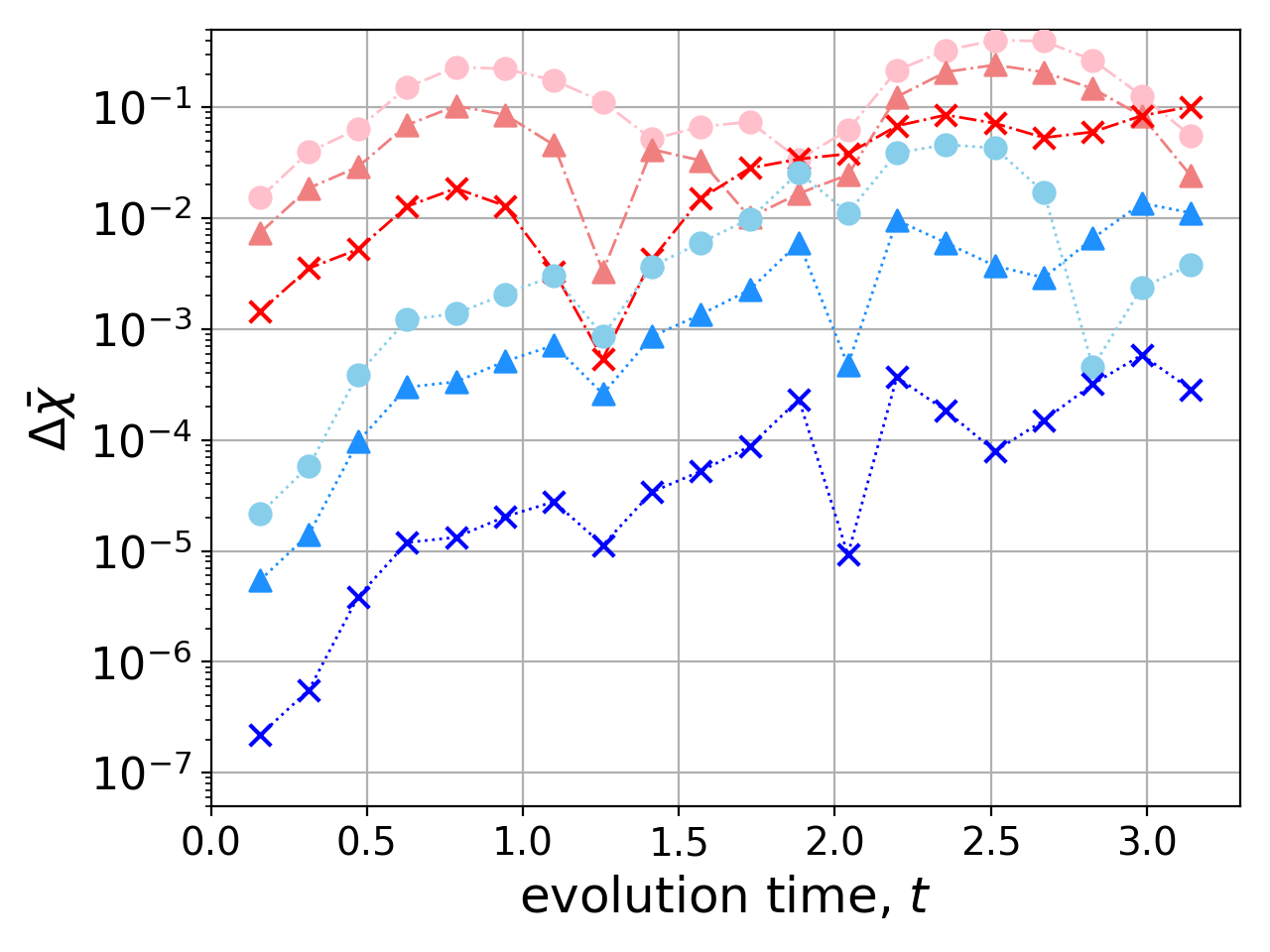}
    }
    \subfloat[]{
        \includegraphics[width=0.48\linewidth]{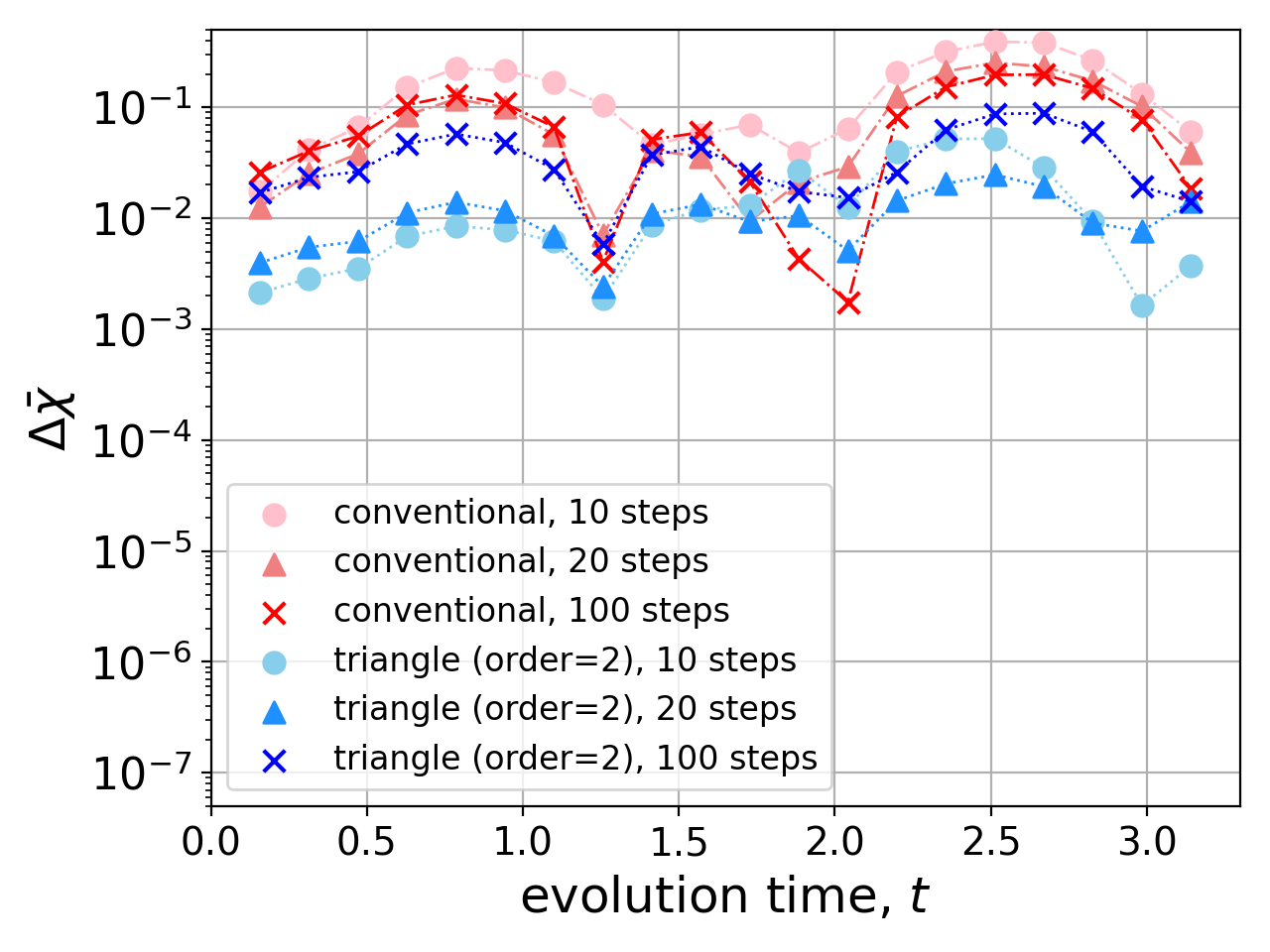}
    }
    \caption{
        Estimated average spin-chirality as a function of the evolution time $t$.
        Figures (a) and (b) show the estimation bias $\Delta \bar{\chi}$, defined as the difference between the estimated and theoretical average spin-chirality, for $t \in \{\pi/20, 2\pi/20, \ldots, \pi\}$.
        For each method, the results for $10$, $20$, and $100$ Trotter steps are indicated by $\bullet$, $\blacktriangle$, and $\times$, respectively.
        The conventional method is shown in shades of red, while the second-order decomposition of the proposed method is shown in shades of blue.
    }
    \label{fig:kagome_ring_chirality}
\end{figure*}

\section{Numerical simulation\label{sec:numerics}}

We illustrate the power of our method by performing the numerical simulation using the 12-qubit $XXX$ Kagome Heisenberg model with three-body spin-chirality interactions within its local triangular sublattices, as discussed in the previous section.
We set the exchange coupling $J=1$ and the spin-chirality coupling $K=0.1$ for the model Hamiltonian.
Throughout the experiments, the initial state of the system is set to the site-dependent product state
\begin{equation}\label{eq:state_initial_kagome}
\begin{split}
    \left|\psi_{\mathrm{init}}\right\rangle
    = \bigotimes_{k=1}^{N}\frac{1}{\sqrt{2}}\left(\left|0\right\rangle + e^{i\frac{2(k-1)\pi}{3}}\left|1\right\rangle\right),
\end{split}
\end{equation}
which corresponds to a phase pattern $(0, 2\pi/3, 4\pi/3)$ across the three-site triangles in the Kagome lattice.
This initial state is prepared from $|0\rangle^{\otimes N}$ state by applying $H$ gate and $RZ(2(k-1)\pi/3)$ gate sequentially for each qubit $k\in[N]$.

We use Qiskit transpiler~\cite{qiskit2024} for circuit optimisation with basis gate set $\left\{RZ, \sqrt{X}, CNOT\right\}$.
Note that $\sqrt{X}$ and CNOT are Clifford gates, and only $RZ$ gates can be non-Clifford.

\subsection{State infidelity with different Trotter steps}

First, we compare the conventional and proposed methods by simulating the time evolution from the same initial state from $t=0$ to $t=\pi$, while varying the number of Trotter steps among $\{10,20,\ldots, 100\}$.
The accuracy of the evolved state is evaluated by the state infidelity
\begin{equation}
\begin{split}
    1 - F(\tilde{\rho}, \rho_{\mathrm{ideal}})
    &= 1 - \operatorname{Tr} \left[\left(\rho_{\mathrm{ideal}}^{1/2} \tilde{\rho} \rho_{\mathrm{ideal}}^{1/2}\right)^{1/2}\right]^{2},
\end{split}
\end{equation}
where $\tilde{\rho}$ is the simulated output state and $\rho_{\mathrm{ideal}}$ is the ideally evolved state.
We investigate the infidelity by the conventional decomposition, the first-order decomposition of the proposed method (referred to as ``triangle'' in the figures), and the second-order decomposition of the proposed method.

We consider two hardware-relevant regimes: a near-term setting, where two-qubit gate noise is the dominant error source, and an early fault-tolerant setting, where non-Clifford gates are the most costly and noise-sensitive operations.
In the near-term setting, we assign the depolarising error with two-qubit depolarising probability $p_{2} \in \{1.0\times 10^{-5}, 1.0\times 10^{-4}\}$ and single-qubit depolarising probability $p_{1}=p_{1}/10$, reflecting noise levels relevant to state-of-the-art and expected near-future quantum hardware~\cite{Hughes2025Trapped-ion, Ransford2025Helios}.
In the early fault-tolerant setting, we assign the dephasing error to non-Clifford $RZ$ gates with probability $p_{z} \in \{1.0\times 10^{-5}, 1.0\times 10^{-4}\}$.

The results for the near-term setting are shown in Fig.~\ref{fig:kagome_ring_infidelity}(a).
Here, we observe that the proposed method yields substantially lower infidelity than the conventional method, for both the number of Trotter steps and the number of CNOT gates.
The second-order decomposition by the proposed method achieves an infidelity more than three orders of magnitude smaller while using only about one-quarter as many CNOT gates.
This indicates that the proposed decomposition suppresses the Trotter error much more efficiently per unit circuit depth.

Under depolarising noise, both methods exhibit the standard trade-off between residual Trotter error and circuit noise accumulation.
A notable feature is that the proposed method reaches its minimum infidelity at a much smaller circuit depth.
This is because its Trotter error decreases so rapidly that, once the approximation error is already sufficiently suppressed, the additional depolarising error becomes dominant.
By contrast, in the conventional method, increasing the number of Trotter steps can still improve overall accuracy, while the Trotter error remains significant over a wider range of depths.

We also observe a similar advantage from Fig.~\ref{fig:kagome_ring_infidelity}(b) for the early fault-tolerant setting.
Note that the circuits are optimised using the default Qiskit transpiler and thus may not be optimised with respect to the number of non-Clifford gates.
Nevertheless, the fact that the proposed method achieves a lower total gate count implies that the number of non-Clifford gates can also be reduced accordingly by our method.
For a more refined comparison, one could further combine our decomposition with compilers tailored to non-Clifford resource optimisation, such as T-count or arbitrary-angle rotation-count reduction.

\subsection{Time evolution of expected average spin-chirality}

We next demonstrate the advantage of the proposed method in the expectation value $\bar{\chi}$ of the average spin-chirality, whose observable is defined as
\begin{equation}
\begin{split}
    O_{\bar{\chi}} = \frac{1}{6}\sum_{j=1}^{6} \vec{\sigma}_{2j}\cdot\left(\vec{\sigma}_{2j-1}\times\vec{\sigma}_{2j+1}\right).
\end{split}
\end{equation}
This observable corresponds to the average of the three-body interactions in the triangular regions of the 12-qubit Kagome Hamiltonian and thus measures how effectively each method handles three-body interactions via its Trotter decomposition.
We evolve the initial state set to Eq.~\eqref{eq:state_initial_kagome} for different evolution time $t\in\{k\pi/20\}_{k=1}^{20}$ with $10$, $20$ and $100$ Trotter steps, respectively.
Note that these Trotter step counts are already used in real-device experiments for simpler physical models~\cite{Chowdhury2024Enhancing}.

We compute the estimation bias $\Delta\bar{\chi}$ between the theoretically predicted expectation value and the one obtained by each Trotter decomposition.
The results for noise-free simulation are plotted in Fig.~\ref{fig:kagome_ring_chirality}.
For the noise-free simulation, Fig.~\ref{fig:kagome_ring_chirality}(a) shows that, compared with the conventional method, the proposed second-order decomposition reduces the estimation bias by more than three orders of magnitude over most of the parameter regime.

In addition to noise-free simulation, we introduce depolarising errors with $p_{1} = 1.0\times 10^{-5}$ and $p_{2} = 10 p_{1} = 1.0\times 10^{-4}$ in the noisy simulation.
The results for noisy simulation are plotted in Fig.~\ref{fig:kagome_ring_chirality}(b).
Observing from Fig.~\ref{fig:kagome_ring_chirality}(b), we see that the proposed method still closely matches the theoretical value, while the plots obtained with the conventional method are substantially altered from it.
This implies that our method outperforms existing approaches in general under depolarising noise, which attenuates expectation values toward zero, and is particularly effective at estimating observables with intrinsically small expectation values with higher precision.

\section{Discussion}

Our main contribution is a new design principle for product formulas that goes beyond commutativity.
Conventional commutativity-based decompositions partition the Hamiltonian into product-formula blocks based on rather direct commutativity criteria.
Although such criteria are universal and straightforward to design, the resulting product formula does not necessarily reflect the intrinsic local structure of the underlying Hamiltonian, which can lead to prohibitively long sequential propagators and large approximation errors.

The proposed design principle is the first to exploit the local Hamiltonian structures beyond commutativity, by taking three-qubit local $\mathrm{SU}(8)$ operators so that they satisfy the underlying $\mathrm{SU}(2)$ symmetry.
This provides concrete evidence that local $\mathrm{SU}(2)$ symmetry can serve as a highly effective clustering criterion for designing a product formula that respects the intrinsic structure of the local dynamics.

Our method can also be interpreted as a ``hyperedge'' decomposition of the Hamiltonian.
From this viewpoint, conventional decompositions largely adopt graph-based partitioning strategies, where non-local physical interactions are organised along edges or paths of an interaction graph.
By contrast, our approach identifies non-trivial local structures shared across multiple terms and sites, moving beyond pair-wise or path-wise commutativity, and groups them into structure-aware local blocks.
This suggests a more flexible route to designing product formula, in which the relevant clustering units are not restricted to ordinary graph edges but may instead extract more physically native local structures of the Hamiltonian.

Both our theoretical and numerical analyses demonstrate that our method substantially enhances the efficiency and scalability of the simulation, reducing both residual errors and circuit overhead across a broad range of demanding, representative many-body systems.
The advantage of our method becomes particularly prominent when many terms in the Hamiltonian representation can be grouped within triangular plaquettes that reveal a local $\mathrm{SU}(2)$ symmetry and can therefore be propagated exactly as a whole.
Summarising the above, our framework not only opens a new avenue for the product formula design but also provides a promising route to the scalable, practical simulation of a broader class of demanding physical models that were previously out of reach.

The proposed design principle opens several directions for future research.
While we have systematically shown how to efficiently implement clustered Hamiltonians once their symmetry classes are identified, determining how to decompose a general Hamiltonian into fewer classes remains an important open problem.
That is, despite going beyond conventional commutativity-based clustering, the search for decompositions that reduce both Trotter error and circuit overhead remains largely heuristic.

Developing tighter residual-error analyses, building on recent techniques for bounding product-formula errors~\cite{Childs2021Theory, Yuan2021Quantum, Watson2025Exponentially}, is also an important direction.
Our current analysis counts the number of non-trivial Pauli terms in the residual error, which provides only a coarse proxy for the actual approximation error and may not fully explain the substantial numerical advantage of the proposed method.
Refining this analysis could reveal broader algebraic conditions for efficient product-formula decompositions beyond the $\mathrm{SU}(2)$ symmetry in the local three-site system studied here.

More broadly, the proposed framework is also naturally compatible with other approaches for improving Trotter-based simulation.
Important directions include incorporating our method into randomised Trotter schemes~\cite{Campbell2019Random, Kiumi2025TE-PAI}, combining it with other quantum error mitigation techniques~\cite{Temme2017Error, Endo2019Mitigating, Hakkaku2025Data-Efficient, yang2022efficient, Yoshioka2022Generalized, Yang2025Resource-Efficient, cai2023quantum}, and optimising it for real-device implementations~\cite{Kim2023Evidence, Yoshioka2025Krylov}.
These directions may further enhance noise robustness and help translate the present framework into more practical experimental advantages.


\begin{acknowledgments}
    N.N. and B.Y. acknowledge fruitful discussions with Nobuyuki Yoshioka, Hiroshi Imai, Tetsuo Shibuya from the University of Tokyo, and Yuichiro Matsuzaki from Chuo University.
\end{acknowledgments}


\appendix

\section{Proof of Theorem~\ref{theorem:classification_of_SU(8)_by_SU(2)}\label{appendix:Theorem_genflame}}

In this section, we prove Theorem~\ref{theorem:classification_of_SU(8)_by_SU(2)} through the following three lemmas. 
To simplify the discussion, we first introduce the notion of the Lie algebra $\mathfrak{su}(d)$.

    A Lie algebra $\mathcal{R}$ is a vector space of operators closed under the commutator operation. 
For two operators $R_{1},R_{2}\in\mathcal{R}$, the binary operations of addition and commutator are defined as
\begin{equation}
    \begin{split}
        \mathcal{R}\times\mathcal{R}&\rightarrow\mathcal{R}\\
        (R_{1},R_{2})&\mapsto R_{1}+R_{2},
    \end{split}
\end{equation}
and
\begin{equation}
    \begin{split}
        \mathcal{R}\times\mathcal{R}&\rightarrow\mathcal{R}\\
        (R_{1},R_{2})&\mapsto \frac{1}{2i}[R_{1},R_{2}],
    \end{split}
\end{equation}
respectively. 

In operator theory, the Lie algebra $\mathcal{R}$ is linearly spanned by a set of operators as
\begin{equation}
    \mathcal{R}=\mathrm{span}_{\mathbb{R}}\{R_{1},R_{2},\dots,R_{N}\},
\end{equation}
with the dimension $N$. In particular, if the unitary operators generated by Hermitian operators $R\in\mathcal{R}$, namely $\exp(-iR)$, form the $\mathrm{SU}(d)$ group, then $\mathcal{R}$ is said to be isomorphic to the $\mathfrak{su}(d)$ algebra:
\begin{equation}
    \mathcal{R}\cong\mathfrak{su}(d).
\end{equation}
The dimension of the $\mathfrak{su}(d)$ algebra $N$ is known as $N=d^{2}-1$.

Next, we define the notion of commutativity between an operator and a Lie algebra, as well as between two Lie algebras. An operator $A$ is said to commute with a Lie algebra $\mathcal{R}$ if
\begin{align}
    [A,R]=0,\quad \text{for all } R\in\mathcal{R}.
\end{align}
Similarly, two Lie algebras $\mathcal{R}$ and $\mathcal{R}'$ are said to commute with each other if
\begin{align}
    [R,R']=0,\quad \text{for all } R\in\mathcal{R},\, R'\in\mathcal{R}'.
\end{align}

The first and second lemmas focus on the fact that the three-qubit cluster Hamiltonian $H_{\mathrm{triangle}}^{(3)}$ can be decomposed into four tractable operators as $H_{\mathrm{triangle}}^{(3)}=H_{1}+H_{2}+H_{3}+H_{4}$, while the third lemma enables us to construct unitary gates that transform a general $\mathfrak{su}(8)$ Hamiltonian $H$ spanned by 63 Pauli strings into the three-qubit cluster Hamiltonian $H_{\mathrm{triangle}}^{(3)}$.

To simplify the proofs of the first and second lemmas without loss of generality, we define the subspace $\mathcal{C} = \mathrm{span}_{\mathbb{R}}\{X_{2}X_{3}, Y_{2}Y_{3}, Z_{2}Z_{3}\}$, which is isomorphic to the three-dimensional torus $\mathcal{T}^{3}$. Here, the algebra of the torus $\mathcal{T}$ indicates that the element $T\in\mathcal{T}$ generates $\mathrm{U}(1)$ operator, that is $\exp(-iT\theta)\cong\mathrm{diag}\bigl(1,e^{i\theta}\bigr)$, and then $d$-dimensional torus $\mathcal{T}^d$ consists $d$ mutually commuting basis elements $\{T_{1},T_{2},\cdots,T_d\}$.

The first lemma aims to determine the number of subspaces $\mathcal{G}_{l} \cong \mathfrak{su}(2)$ whose intersection with $\mathcal{C}$ is trivial (i.e., consists only of the zero element), as illustrated in Fig.~\ref{fig:Upper_bound}.

\begin{lemma}
  Given the conditions that
  $\mathcal{G}_{l}\cap \mathcal{G}_{l^{\prime}}=\{0\}$
  for any $l\neq l^{\prime}$ and that each
  $\mathcal{G}_{l}$ commutes with $\mathcal{C}$,
  the number of subspaces $\mathcal{G}_{l}$ is four.
  These subspaces satisfy
  \begin{equation}
      \begin{split}
      \bigoplus_{l=1}^{4}
    \mathcal{G}_{l}
    =
    \mathrm{span}_{\mathbb{R}}
    \bigl\{
        GC
        \,\big|\,
        G\in\{X_{1},Y_{1},Z_{1}\},
        \,
        C\in\left(\mathcal{C}\cup\{I\}\right)
    \bigr\}.
      \end{split}
            \label{eq:GEN}
  \end{equation}
\label{lemma:1}
\end{lemma}

\begin{proof}
There exists an operator $C\in\mathcal{C}$ that does not commute with each of the following single-qubit operators acting on the second and third qubits:
\begin{equation}
    \{ X_{2}, Y_{2}, Z_{2}, X_{3}, Y_{3}, Z_{3} \},
\label{eq:one_ope_{2}3}
\end{equation}
and the following two-qubit operators:
\begin{equation}
    \{ X_{2} Y_{3}, Y_{2} Z_{3}, Z_{2} X_{3},
    Y_{2} X_{3}, Z_{2} Y_{3}, X_{2} Z_{3} \}.
\label{eq:two_ope_{2}3}
\end{equation}
For example,
\begin{equation}
    [ X_{2}, Y_{2} Y_{3} ]
    =
    2 i Z_{2} Y_{3}
    \neq 0.
\end{equation}

Therefore, an operator commuting with $\mathcal{C}$ must be written as $GC$, where
$G\in\{X_{1},Y_{1},Z_{1},I\}$ and
$C\in(\mathcal{C}\cup \{I\})$.
Excluding the case $G=I$, which corresponds to operators already contained in $\mathcal{C}$, there remain twelve non-trivial operators, given by the right-hand side of Eq.~\eqref{eq:GEN}.

Since each subspace $\mathcal{G}_{l}\cong\mathfrak{su}(2)$ has dimension three and the intersections are trivial, the number of subspaces is exactly four to cover the twelve non-trivial operators in Eq.~\eqref{eq:GEN}.
\end{proof}

We remark that the basis of $\mathcal{C}$ can be reconstructed from the elements of $\mathcal{G}_{l}$. An element $C_\mu\in\mathcal{C}$ can be given by the product between two elements $G_\mu\in\mathcal{G}_{l}$ and $G_\mu'\in\mathcal{G}_{l^{\prime}}$ that commute with each other:
\begin{equation}
    C_\mu=G_\mu G_\mu' \label{eq:remark},
\end{equation}
holds for all $l\not=l^{\prime}$, e.g., $(G_{1},G_{2},G_{3})=(X_{1},Y_{1},Z_{1})$ and $(G_{1}^{\prime},G_{2}^{\prime},G_{3}^{\prime})=(X_{1}X_{2}X_{3},Y_{1}Y_{2}Y_{3},Z_{1}Z_{2}Z_{3})$. 
This remark is helpful for the proof of the following second lemma.

In the proof of the second lemma, we choose the bases of $\mathcal{G}_{1}$ and $\mathcal{G}_{2}$ as
\begin{equation}
\begin{split}
    \mathcal{G}_{1} &= \mathrm{span}_{\mathbb{R}}\{X_{1},\, Y_{1},\, Z_{1} \},\\
        \mathcal{G}_{2} &= \mathrm{span}_{\mathbb{R}}\{X_{1} X_{2} X_{3},\, Y_{1} Y_{2} Y_{3},\, Z_{1} Z_{2} Z_{3}\}.
\end{split}
\label{eq:GEN1}
\end{equation}
\begin{lemma}
    Assuming that the subspace $\mathcal{H}_{l}$ commutes with $\mathcal{G}_{l}$, the intersection of $\mathcal{H}_{l}$ and $\mathcal{H}_{l^{\prime}}$ for any $l\neq l^{\prime}$ satisfies
    \begin{equation}
        \mathcal{H}_{l}\cap\mathcal{H}_{l^{\prime}}=\mathcal{C}=\mathrm{span}_{\mathbb{R}}\left\{X_{2}X_{3},Y_{2}Y_{3},Z_{2}Z_{3}\right\}.
    \end{equation}
\label{lemma:2}
\end{lemma}

\begin{proof}
Since any operator $C\in\mathcal{C}$ commutes with all subspaces $\mathcal{G}_{l}$, we have $C\in\mathcal{H}_{l}$ for all $l$, implying  $\mathcal{C}\subset\bigcap_{l=1}^{4}\mathcal{H}_{l}$.

Next, we assume that there exists an operator $A\in\mathcal{H}_{l}\cap\mathcal{H}_{l^{\prime}}$ for $l^{\prime} \neq l$, but $A\not\in\mathcal{C}$, and derive a contradiction. For simplicity, we apply a unitary transformation that maps $\mathcal{G}_{l}$ to $\mathcal{G}_{1}$ and $\mathcal{G}_{l^{\prime}}$ to $\mathcal{G}_{2}$ without changing the basis of $\mathcal{C}$, based on the reconstruction of  $\mathcal{C}$ in Eq.~\eqref{eq:remark}.

The operator $A$ is transformed to another operator $B$ that also satisfies $B\notin\mathcal{C}$ and acts on the second or third qubits. However, Lemma~\ref{lemma:1} illustrates that such an operator $B$ does not commute with $\mathcal{C}$. We subsequently have the consequence that the operator $B$ does not commute with the subspace $\mathcal{G}_{2}$, yielding $B\not\in\mathcal{H}_{2}$, which contradicts the assumption. Therefore, we yield $A\not\in \mathcal{H}_l\cap\mathcal{H}_{l'}$, and thus we conclude $\mathcal{H}_{l}\cap\mathcal{H}_{l^{\prime}}=\mathcal{C}$.
\end{proof}

Owing to Lemma~\ref{lemma:1} and Lemma~\ref{lemma:2}, any pairs of two subspaces
$\mathcal{H}_{l}\oplus\mathcal{G}_{l}$ ($l=1,\cdots,4$) intersects with only $\mathcal{C}$ for each $l$. The dimension of  $\mathcal{H}_{l}\oplus\mathcal{G}_{l}\cong\mathfrak{su}(4)\oplus\mathfrak{su}(2)$ is equal to $15+3=18$. Subtracting the dimension of their intersection $\mathcal{C}$, the number of independent elements outside $\mathcal{C}$ is 
$(18-3)\times4 = 60$. Therefore, the total number of elements is $60+3=63$, which exactly matches the dimension of the $\mathfrak{su}(8)$ algebra. We conclude that four classes are sufficient to represent all three-qubit cluster Hamiltonians $H_{\mathrm{triangle}}^{(3)}$.

We finally prove the third lemma to minimise the subspace that generates the full three-qubit $\mathfrak{su}(8)$ algebra. This lemma plays a crucial role in determining a unitary transformation that maps the given $\mathfrak{su}(8)$ Hamiltonian $H$ to the three-qubit cluster Hamiltonian $H_{\mathrm{triangle}}^{(3)}$.

\begin{lemma}
Let $\mathcal{G}_{1}\cong\mathfrak{su}(2)$ and 
$\mathcal{C}\cong\mathcal{T}^3$, and let $\mathcal{G}_{2}$ be spanned by elements obtained in Eq.~\eqref{eq:remark}. 
Furthermore, we define $\mathcal{S}\subset\mathcal{H}_{1}$ as an additional 
$\mathfrak{su}(2)$ subalgebra satisfying $\mathcal{S}\cap\mathcal{C}=\{0\}$.
Then the Lie algebra generated by these four subspaces is unitarily equivalent to the standard three-qubit $\mathfrak{su}(8)$ algebra.
\label{lemma:22}
\end{lemma}
\begin{proof}
There exists a unitary operator $U$ that realises the following mappings:
\begin{equation}
U:\begin{cases}
    \mathcal{G}_{1}&\rightarrow    \quad\mathcal{G}_{1}^{\prime}:=\mathrm{span}_{\mathbb{R}}\bigl\{X_{1},Y_{1},Z_{1}\bigr\}    \\
    \mathcal{S}&\rightarrow\quad\mathcal{S}_{1}:=    \mathrm{span}_{\mathbb{R}}\bigl\{
    X_{2},Y_{2},Z_{2}\bigr\}    \\
    \mathcal{C}&\rightarrow \quad\mathcal{C}^{\prime}:=    \mathrm{span}_{\mathbb{R}}\bigl\{
    X_{2}X_{3},Y_{2}Y_{3},Z_{2}Z_{3}\bigr\}
\end{cases} .       
\label{eq:UF}
\end{equation}
Under this transformation, $\mathcal G_{2}$ is mapped as
\begin{equation}
\begin{split}
    U:&\mathcal{G}_{2}=\mathrm{span}_{\mathbb{R}}\{G_{1}C_{1},G_{2}C_{2},G_{3}C_{3}\}\\
    \rightarrow&\mathcal{G}_{2}^{\prime}:=\mathrm{span}_{\mathbb{R}}\{X_{1}X_{2}X_{3},Y_{1}Y_{2}Y_{3},Z_{1}Z_{2}Z_{3}\}.
\end{split}
\end{equation}

Next, we prove that the transformed subspaces generate the three-qubit $\mathfrak{su}(8)$ algebra. The subspace $\mathcal{S}_{1}\oplus\mathcal{C}^{\prime}$ generates the $\mathfrak{su}(4)$ algebra acting on the second or third qubits. The commutators between $\mathcal{G}_{1}^{\prime}$ and $\mathcal{G}_{2}^{\prime}$ then generate the additional $\mathfrak{su}(2)$ subspaces $\mathcal{G}_{3}^{\prime}:=\mathrm{span}_{\mathbb{R}}\{X_{1}Y_{2}Y_{3},Y_{1}Z_{2}Z_{3},Z_{1}X_{2}X_{3}\}$ and $\mathcal{G}_{4}^{\prime}:=\mathrm{span}_{\mathbb{R}}\{X_{1}Z_{2}Z_{3},Y_{1}X_{2}X_{3},Z_{1}Y_{2}Y_{3}\}$. We further define two $\mathfrak{su}(2)$ subspaces:
    \begin{equation}
    \begin{split}
        \mathcal{S}_{2}&=\mathrm{span}_{\mathbb{R}}\{X_{1}Y_{2}Z_{3},Y_{1}Z_{2}X_{3},Z_{1}X_{2}Y_{3}\},\\
        \mathcal{S}_{3}&=\mathrm{span}_{\mathbb{R}}\{X_{1}X_{2}Z_{3},Y_{1}Y_{2}X_{3},Z_{1}Z_{2}Y_{3}\},
    \end{split}
    \label{eq:S2S3}
    \end{equation}
    where the six operators in Eq.~\eqref{eq:S2S3} are generated by the commutator between $\mathcal{G}_{4}^{\prime}$ and $\mathcal{S}_{1}$, e.g.
    \begin{equation}
    [X_{1}Z_{2}Z_{3},X_{2}]=2iX_{1}Y_{2}Z_{3}\in\mathcal{S}_{2}.    
    \end{equation} 
    The subspaces $\mathcal{S}_l\oplus\mathcal{C}'$ $(l=2,3)$ generate $\mathcal{H}_{l}\cong\mathfrak{su}(4)$ algebra through commutators. 
    
    Similarly, we can obtain another $\mathfrak{su}(2)$ algebra $\mathcal{S}_{4}=\mathrm{span}_{\mathbb{R}}\{X_{1}Y_{2}X_{3},Y_{1}Z_{2}Y_{3},Z_{1}Y_{2}Z_{3}\}$ generated by commutator between $\mathcal{G}_{4}^{\prime}$ and $\mathcal{S}_{1}$. The subspace $\mathcal{S}_4\oplus\mathcal{C}'$ also generates $\mathcal{H}_{4}\cong\mathfrak{su}(4)$ through commutators.
    
    Therefore, we conclude that the four subspaces $\mathcal{G}_{1}$, $\mathcal{G}_{2}$, $\mathcal{S}_{1}$, and $\mathcal{C}'$ span $\sum_{i=1}^{4}\mathcal{G}_{i}^{\prime}\oplus\mathcal{H}_{i}$ corresponding to the three-qubit $\mathfrak{su}(8)$ algebra. This proves that the Lie algebra generated by $\mathcal{G}_{1}\oplus\mathcal{G}_{2}\oplus\mathcal{S}\oplus\mathcal{C}$ is unitarily equivalent to the three-qubit $\mathfrak{su}(8)$ algebra.
\end{proof}

We next construct a unitary transformation $U$ that maps
$\mathcal G_{1}$, $\mathcal S$, and $\mathcal C$
to the local three-qubit subspace in Eq.~\eqref{eq:UF} without considering the transformations of all 63 $\mathfrak{su}(8)$ elements. We first consider a unitary transformation $U_{1}$ such that
\begin{equation}
    U_{1}:\mathcal{G}_{1}\rightarrow\mathcal{G}_{1}^{\prime}:=\mathrm{span}_{\mathbb{R}}\{X_{1},Y_{1},Z_{1}\}.
\end{equation}
An operator $O\in\mathcal{H}_{1}$ commuting with $\mathcal{G}_{1}$ acts trivially on the first qubit. There exists, therefore, a second unitary transformation $U_{2}$ satisfying
\begin{equation}
U_{2}\circ U_{1}:\mathcal{S}\rightarrow\mathcal{S}_{1}:=\mathrm{span}_{\mathbb{R}}\{X_{2},Y_{2},Z_{2}\},      \end{equation}
and
\begin{equation}
U_2:\mathcal{G}_{1}^{\prime}\rightarrow\mathcal{G}_{1}^{\prime}    .
\end{equation}
 Moreover, from Lemma ~\ref{lemma:1}, any operator in $\mathcal{H}_{1}$ does not commute with $\mathcal{C}$ except $\mathcal{C}$ itself. Hence, there exists a third unitary transformation $U_{3}$ as
 \begin{equation}
U_{3}\circ U_{2}\circ U_{1}:\mathcal{C}\rightarrow\mathcal{C}^{\prime}:=\mathrm{span}_{\mathbb{R}}\{X_{2}X_{3},Y_{2}Y_{3},Z_{2}Z_{3}\},
\end{equation}
which trivially acts on the first and second qubit. Therefore, the transformation $U=U_{3}U_{2}U_{1}$ realises the desired mapping. 

We remark that the unitary transformations $U_{1}$, $U_{2}$, and $U_{3}$ can be constructed using only Clifford gates. For instance, the subspace $\mathcal{G}_{1}$ is spanned by $\{X_{i}W,Y_{i}W',Z_{i}WW'\}$ where $W$ and $W'$ commute with each other. By reordering Pauli strings using SWAP gates and applying local Clifford operations such as $\sqrt{X}$ and $\sqrt{Z}$, this can be simplified. For example, we obtain
\begin{equation}
\begin{split}
    U:\mathcal{G}_{1}
    &\rightarrow
        \mathrm{span}_{\mathbb{R}}
        \bigl\{
            X_{1}X_{2}\cdots X_{2m+1}X_{2m+2}\cdots X_{2m+m'+1},\\
            &\quad\quad\quad\quad~ Y_{1}Y_{2}\cdots Y_{2m+1}X_{2m+2}\cdots X_{2m+m'+1},\\
            &\quad\quad\quad\quad~ Z_{1}Z_{2}\cdots Z_{2m+1}
        \bigr\}.
\end{split}
\end{equation}
Then, using only CNOT gates, multi-qubit Pauli operators can be transformed into single-qubit operators such as $X_{1}$ and $Z_{1}$.
In the same way, $\mathcal{C}\oplus\mathcal{S}$ can be unitarily transformed into a two-qubit $\mathfrak{su}(4)$ algebra acting on the second or third qubits.

\section{Completeness of Algebraic structure in Hamiltonian partitioned into three tractable blocks\label{appendix:Theorem_complete}}

The three-qubit propagator $\exp(-iH \Delta t)$ may be decomposed into three Trotter blocks based on another algebraic structure mentioned in Theorem~\ref{theorem:classification_of_SU(8)_by_SU(2)}.
To decompose the propagator into three Trotter blocks, we prove the following algebraic theorem:

\begin{theorem}
 When we decompose the $\mathrm{SU}(8)$ propagator into three Trotter blocks, the subspace $\mathcal{V}^{\prime}$ is completely classified as
 \begin{equation}
\mathcal{V}^{\prime}=\sum_{l=1}^{3}\mathcal{H}_{l}\oplus\mathcal{G}_{l}.\label{eq:Omega'}
 \end{equation}
For any pairs  $l\not=l^{\prime}$, the subspaces $\mathcal{H}_{l}\cong\mathfrak{su}(4)$ and $\mathcal{G}_{l}\cong\mathfrak{su}(2)$ satisfy one of the following conditions:\newline
 (i) 
 \begin{equation}
     \mathcal{H}_{l}\cap\mathcal{H}_{l^{\prime}}\cong\mathfrak{su}(2)\,\, and\,\,
     \mathcal{H}_{l}\cap\mathcal{G}_{l^{\prime}}\not=\left\{0\right\}\label{eq:cond1_HG},
 \end{equation}
  (ii)
 \begin{equation}
     \begin{cases}
         \mathcal{H}_{l}\cap\mathcal{H}_{l^{\prime}}\cong\mathfrak{su}(2)\,\, and \,\,\mathcal{H}_{l}\cap\mathcal{G}_{l^{\prime}}\not=\left\{0\right\}&{\rm for\,\,}\{l,l^{\prime}\}=\{1,2\}\\
     \mathcal{H}_{l}\cap\mathcal{H}_{l^{\prime}}\cong \mathcal{T}^{3}\,\, and \,\,  \mathcal{H}_{l}\cap\mathcal{G}_{l^{\prime}}=\left\{0\right\}& {\rm otherwise},
     \end{cases}
     \label{eq:cond2HG}
 \end{equation}
 (iii)
 \begin{equation}
     \begin{cases}
         \mathcal{H}_{l}\cap\mathcal{H}_{l^{\prime}}\cong \mathcal{T}^{3}\,\, and \,\,\mathcal{H}_{l}\cap\mathcal{G}_{l^{\prime}}=\left\{0\right\}&{\rm for\,\,}\{l,l^{\prime}\}=\{1,2\}\\
     \mathcal{H}_{l}\cap\mathcal{H}_{l^{\prime}}\cong \mathfrak{su}(2)\,\, and \,\,  \mathcal{H}_{l}\cap\mathcal{G}_{l^{\prime}}\not=\left\{0\right\} &{\rm otherwise},
     \end{cases}
      \label{eq:cond3HG}
 \end{equation}
  (iv)
 \begin{equation}
     \mathcal{H}_{l}\cap\mathcal{H}_{l^{\prime}}\cong \mathcal{T}^{3}\,\, and\,\,
     \mathcal{H}_{l}\cap\mathcal{G}_{l^{\prime}}=\left\{0\right\}.\label{eq:cond4_HG}
 \end{equation}
 \label{theorem:2}
\end{theorem}

We first present two lemmas that characterise the intersections between the subspaces $\mathcal{H}_{l} \oplus \mathcal{G}_{l}$ and $\mathcal{H}_{l^{\prime}} \oplus \mathcal{G}_{l^{\prime}}$ for $l \neq l^{\prime}$. These lemmas distinguish between the cases in which the intersection $\mathcal{H}_{l} \cap \mathcal{G}_{l^{\prime}}$ is trivial or non-trivial, leading to different consequences for the intersections $\mathcal{H}_{l} \cap \mathcal{H}_{l^{\prime}}$ and $\mathcal{G}_{l} \cap \mathcal{H}_{l^{\prime}}$.

\begin{lemma}
The trivial intersection $\mathcal{H}_{l}\cap\mathcal{G}_{l^{\prime}}=\{0\}$ results in 
\begin{equation}
\begin{split}
    \mathcal{G}_{l}\cap\mathcal{H}_{l^{\prime}}&=\{0\}\\
        \mathcal{H}_{l}\cap\mathcal{H}_{l^{\prime}}&\cong \mathcal{T}^{3},
\end{split}
\end{equation}
where basis elements of the three-dimensional torus $\mathcal{T}^{3}$ are given by Eq.~\eqref{eq:remark}.
\label{lemma:3}
\end{lemma}
\begin{proof}
    We set $\mathcal{G}_{l}=\mathrm{span}_{\mathbb{R}}\left\{X_{1},Y_{1},Z_{1}\right\}$ without loss of generality. Once $\mathcal{G}_{l}$ is determined, the basis of $\mathcal{H}_{l}$ is determined as the set of all operators acting only on second or third qubits.

    Under the condition of $\mathcal{H}_{l}\cap\mathcal{G}_{l^{\prime}}=\{0\}$, the subspace $\mathcal{G}_{l^{\prime}}$ must take the form
    \begin{equation}
        \mathcal{G}_{l^{\prime}}=\mathrm{span}_{\mathbb{R}}\left\{X_{1} C_{1},Y_{1} C_{2},Z_{1} C_{1}C_{2}\right\},
    \end{equation}
    where $C_{1}$ and $C_{2}$ act only on the second or third qubits and commute with each other, but $C_{1}\not=C_{2}$; any other representation of $C_{l}$'s yields an operator not acting on the second and third qubits, which belong to $\mathcal{H}_{l}$. Using Eq.~\eqref{eq:remark} with
    \begin{equation}
    \begin{split}
        (G_{1},G_{2},G_{3})&=(X_{1},Y_{1},Z_{1}),\\
        (G_{1}^{\prime},G_{2}^{\prime},G_{3}^{\prime})&=(X_{1}C_{1},Y_{1}C_{2},Z_{1}C_{1}C_{2}),
    \end{split}
    \end{equation}
    we find that the subspace $\mathcal{C}=\mathrm{span}_{\mathbb{R}}\left\{C_{1},C_{2},C_{1}C_{2}\right\}$ commutes with all elements of $\mathcal{G}_{l}$ and $\mathcal{G}_{l^{\prime}}$. 
    
        Furthermore, any operator $G_\mu\in\mathcal{G}_{l}$ does not commute with $\mathcal{G}_{l^{\prime}}$, indicating that $G_\mu\not\in\mathcal{H}_{l^{\prime}}$ and thus $\mathcal{G}_{l}\cap\mathcal{H}_{l^{\prime}}=\{0\}$. From the proof of Lemma~\ref{lemma:1}, we yield $\mathcal{H}_{l}\cap\mathcal{H}_{l^{\prime}}=\mathcal{C}\cong\mathcal{T}^3$.
\end{proof}

We next prove the lemma that describes another type of intersection between $\mathcal{H}_{l}$ and $\mathcal{H}_{l^{\prime}}$.

\begin{lemma}
The non-trivial intersection $\mathcal{H}_{l}\cap\mathcal{G}_{l^{\prime}}\neq\{0\}$ results in 
\begin{equation}
\begin{split}
    \mathcal{G}_{l}\cap\mathcal{H}_{l^{\prime}}&\neq\{0\}\\
        \mathcal{H}_{l}\cap\mathcal{H}_{l^{\prime}}&\cong \mathfrak{su}(2),
\end{split}
\end{equation}
where the basis set of the $\mathfrak{su}(2)$ algebra is uniquely determined.
\label{lemma:4}
\end{lemma}
\begin{proof}
    We again set $\mathcal{G}_{l}=\mathrm{span}_{\mathbb{R}}\left\{X_{1},Y_{1},Z_{1}\right\}$ and assume that $\mathcal{G}_{l}\cap\mathcal{H}_{l^{\prime}}$ contains the element $X_{1}$.

    We first show that $\mathcal{H}_{l}\cap\mathcal{G}_{l^{\prime}}\not=\{0\}$. The condition $X_{1}\in\mathcal{H}_{l^{\prime}}$ implies that the operator $X_{1}$ commutes with $\mathcal{G}_{l^{\prime}}$. In this case, the basis of $\mathcal{G}_{l^{\prime}}$ must take one of the following forms:
    \begin{equation}
        \mathcal{G}_{l^{\prime}}=\mathrm{span}_{\mathbb{R}}\left\{X_{1}C_{1}^{\prime},X_{1}C_{2}^{\prime},C_{3}^{\prime}\right\}\cong\mathfrak{su}(2),
    \end{equation}
    or
    \begin{equation}
        \mathcal{G}_{l^{\prime}}=\mathrm{span}_{\mathbb{R}}\left\{C_{1}^{\prime},C_{2}^{\prime},C_{3}^{\prime}\right\}\cong\mathfrak{su}(2),\label{eq:Glprim3}
    \end{equation}
    where $C_{1}^{\prime}$, $C_{2}^{\prime}$, and $C_{3}^{\prime}$ act only on the second or third qubits. In both cases, we have $C_{3}^{\prime}\in\mathcal{H}_{l}$.

    We set $(C_{1}^{\prime},C_{2}^{\prime},C_{3}^{\prime})=(X_{2},Y_{2},Z_{2})$ without loss of generality. Under this choice, any operator commuting with $\mathcal{G}_{l}\oplus\mathcal{G}_{l^{\prime}}$ must act only on the third qubit.
    Therefore, we uniquely obtain $\mathcal{H}_{l}\cap\mathcal{H}_{l^{\prime}}=\mathrm{span}_{\mathbb{R}}\{X_{3},Y_{3},Z_{3}\}$.

    The same argument applies if the intersection $\mathcal{G}_{l}\cap\mathcal{H}_{l^{\prime}}$ contains $Y_{1}$ or $Z_{1}$ instead of $X_{1}$.
\end{proof}

Let us prove Theorem~\ref{theorem:2}. We define the subspace $\mathcal{V}^{\prime}$ as
\begin{equation}
    \mathcal{V}^{\prime}=\sum_{l=1}^{3}\left(\mathcal{H}_{l}\oplus\mathcal{G}_{l}\right).
\end{equation}
Based on Lemmas~\ref{lemma:3} and~\ref{lemma:4}, we define two conditions $\mathcal{P}_{ll^{\prime}}$ and $\mathcal{Q}_{ll^{\prime}}$ for a three-qubit cluster Hamiltonian $H\in\mathcal{V}^{\prime}$ as
\begin{equation}
\begin{split}
    \mathcal{P}_{ll^{\prime}}&=\{H\in\mathcal{V}^{\prime} \mid \mathcal{H}_{l}\cap\mathcal{G}_{l^{\prime}}\neq\{0\}\},\\
    \mathcal{Q}_{ll^{\prime}}&=\{H\in\mathcal{V}^{\prime} \mid \mathcal{H}_{l}\cap\mathcal{G}_{l^{\prime}}=\{0\}\},
\end{split}
\end{equation}
for $l\neq l^{\prime}$. Using the distributive law of sets, $\mathcal{V}^{\prime}$ can be rewritten as
\begin{equation}
\begin{split}
 \mathcal{V}^{\prime}&=\bigcap_{l=1}^{3}\bigcap_{l^{\prime}>l}^{3}(\mathcal{P}_{ll^{\prime}}\cup \mathcal{Q}_{ll^{\prime}})\\
    &=(\mathcal{P}_{12}\cap \mathcal{P}_{13}\cap \mathcal{P}_{23})\\
    &\cup(\mathcal{Q}_{12}\cap \mathcal{P}_{13}\cap \mathcal{P}_{23})\cup(1\leftrightarrow2)\cup(1\leftrightarrow3)\\
    &\cup(\mathcal{Q}_{12}\cap \mathcal{Q}_{13}\cap \mathcal{P}_{23})\cup(1\leftrightarrow2)\cup(1\leftrightarrow3)\\
    &\cup(\mathcal{Q}_{12}\cap \mathcal{Q}_{13}\cap \mathcal{Q}_{23}).
\end{split}
\label{eq:distributeset}
\end{equation}
Focusing only on the number of occurrences of $\mathcal{P}_{ll^{\prime}}$, we obtain the following four patterns:

\begin{itemize}
    \item
    Third order of $\mathcal{P}_{ll^{\prime}}$

From Lemma~\ref{lemma:4}, we obtain $\mathcal{G}_{l}\cap\mathcal{H}_{l^{\prime}}\neq\{0\}$ for all pairs with $l\neq l^{\prime}$. In addition, any intersection between $\mathcal{H}_{l}$ and $\mathcal{H}_{l^{\prime}}$ is isomorphic to an $\mathfrak{su}(2)$ algebra. This condition corresponds to condition (i) in Theorem~\ref{theorem:2}.
    
    \item Second order of $\mathcal{P}_{ll^{\prime}}$

We consider only the set $\mathcal{Q}_{12}\cap \mathcal{P}_{13}\cap \mathcal{P}_{23}$ owing to the exchange symmetry of the indices in Eq.~\eqref{eq:distributeset}. 

The non-trivial intersections $\mathcal{G}_{3}\cap\mathcal{H}_{l}\neq\{0\}$ for $l\neq3$ implies $\mathcal{H}_{3}\cap\mathcal{G}_{l}\neq\{0\}$ by Lemma~\ref{lemma:4}. Furthermore, any intersection between $\mathcal{H}_{3}$ and $\mathcal{H}_{l}$ for $l\neq3$ is isomorphic to an $\mathfrak{su}(2)$ algebra.

On the other hand, the trivial intersection $\mathcal{G}_{1}\cap\mathcal{H}_{2}=\{0\}$ implies $\mathcal{H}_{2}\cap\mathcal{G}_{1}=\{0\}$ by Lemma~\ref{lemma:3}. Moreover, $\mathcal{H}_{1}\cap\mathcal{H}_{2}$ is isomorphic to $\mathcal{T}^{3}$. This condition corresponds to condition (iii) in Theorem~\ref{theorem:2}.

    \item First order of $\mathcal{P}_{ll^{\prime}}$

We consider only the set $\mathcal{P}_{12}\cap \mathcal{Q}_{13}\cap \mathcal{Q}_{23}$ in Eq.~\eqref{eq:distributeset}. 

The trivial intersections $\mathcal{G}_{3}\cap\mathcal{H}_{l}=\{0\}$ for $l\neq3$ implies $\mathcal{H}_{3}\cap\mathcal{G}_{l}=\{0\}$ by Lemma~\ref{lemma:3}. Furthermore, any intersection between $\mathcal{H}_{3}$ and $\mathcal{H}_{l}$ for $l\neq3$ is isomorphic to $\mathcal{T}^3$.

On the other hand, the non-trivial intersection $\mathcal{G}_{1}\cap\mathcal{H}_{2}\neq\{0\}$ implies $\mathcal{H}_{2}\cap\mathcal{G}_{1}\neq\{0\}$ by Lemma~\ref{lemma:4}. Moreover, $\mathcal{H}_{1}\cap\mathcal{H}_{2}$ is isomorphic to $\mathfrak{su}(2)$. This condition corresponds to condition (ii) in Theorem~\ref{theorem:2}.

    \item Zeroth order of $\mathcal{P}_{ll^{\prime}}$

From Lemma~\ref{lemma:3}, we obtain $\mathcal{G}_{l}\cap\mathcal{H}_{l^{\prime}}=\{0\}$ for all pairs with $l\neq l^{\prime}$. In addition, any intersection between $\mathcal{H}_{l}$ and $\mathcal{H}_{l^{\prime}}$ is isomorphic to $\mathcal{T}^{3}$. This condition corresponds to condition (iv) in Theorem~\ref{theorem:2}.

\qedsymbol
\end{itemize}

We finally remark that algebraic representations of the subspace $\mathcal{V}^{\prime}$ can be classified more strictly by the following corollary.

\begin{corollary}
Let $\{G_{1},G_{2},G_{3}\}$ be a basis of $\mathcal{G}_{l}\cong\mathfrak{su}(2)$ for a fixed $l$. 
If $\mathcal{H}_{l^{\prime}}\cap\mathcal{G}_{l}\neq\{0\}$ for any $l^{\prime}\neq l$, then the intersections fall into the following four patterns, where $l^{\prime\prime}\not\in\{l,l^{\prime}\}$:
    \begin{itemize}
\item[(A)]
$\mathcal{H}_{l^{\prime}}\cap\mathcal{G}_{l}=\{G_{1}\}$ and $\mathcal{H}_{l^{\prime\prime}}\cap\mathcal{G}_{l}=\{G_{1}\}$, 
    
    \item[(B)] $\mathcal{H}_{l^{\prime}}\cap\mathcal{G}_{l}=\{G_{1}\}$ and $\mathcal{H}_{l^{\prime\prime}}\cap\mathcal{G}_{l}=\{G_{2}\}$, but $(\mathcal{H}_{l^{\prime}}\cap\mathcal{H}_{l^{\prime\prime}})\cap\mathcal{G}_{l}=\{0\}$, 
    
    \item[(C)] $(\mathcal{H}_{l^{\prime}}\cap\mathcal{H}_{l^{\prime\prime}})\cap\mathcal{G}_{l}=\mathcal{G}_{l}$,

    \item[(D)] $\mathcal{H}_{l^{\prime}}\cap\mathcal{G}_{l}=\mathcal{G}_{l}$ and $\mathcal{H}_{l^{\prime\prime}}\cap\mathcal{G}_{l}=\{G_{1}\}$.
    
    \end{itemize}
\label{corollary:1}
\end{corollary}

\begin{figure*}
    \centering
    \includegraphics[width=0.8\linewidth]{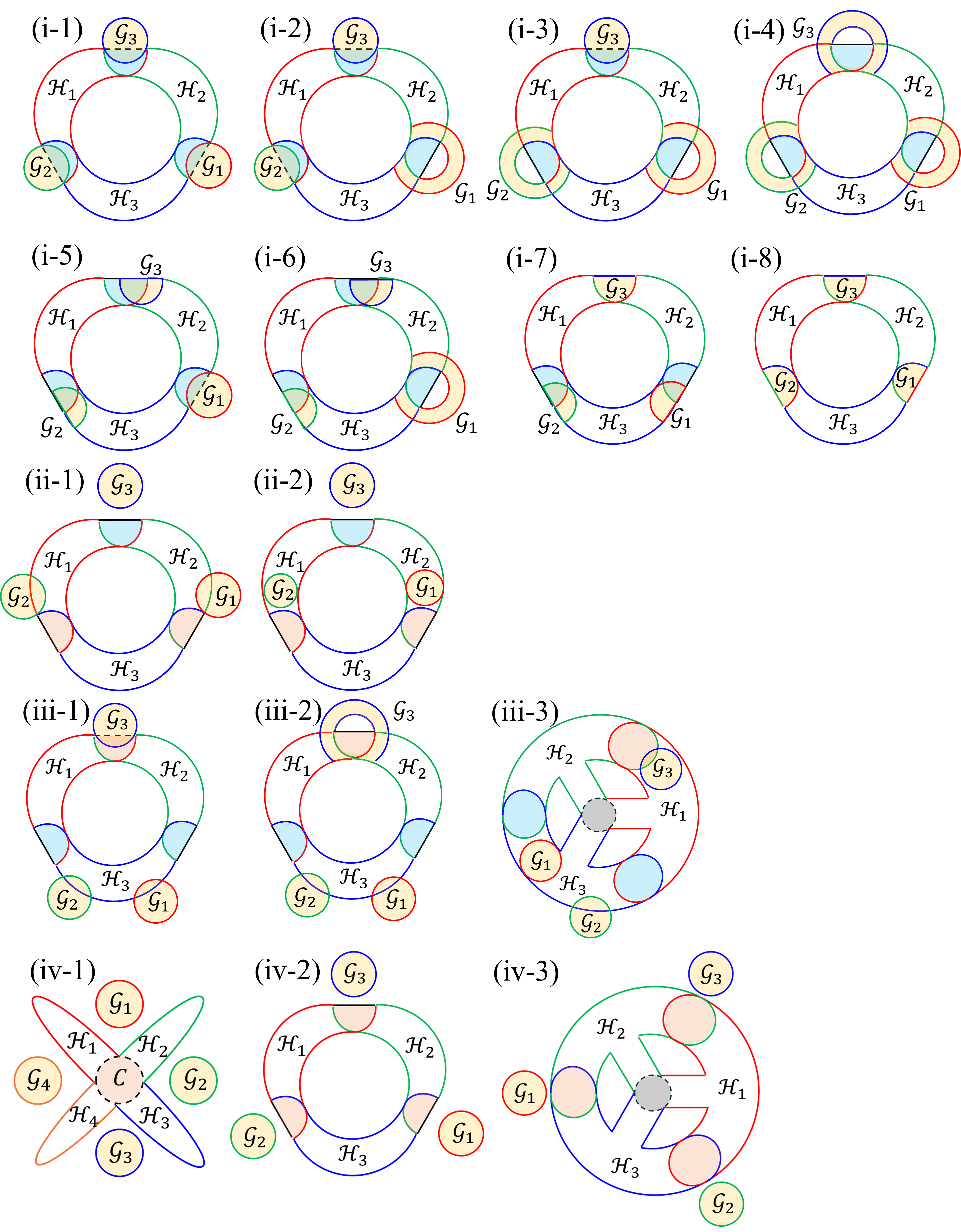}
    \caption{
        The complete algebraic structures in Theorem~\ref{theorem:2}. The yellow transparent regions denote $\mathcal{G}_{l}\cong\mathfrak{su}(2)$.
        The light blue and orange transparent regions denote the $\mathfrak{su}(2)$ algebra and the three-dimensional torus $\mathcal{T}^{3}$, respectively.
        The gray transparent regions represent one element belonging to $\mathcal{H}_{1}\cap\mathcal{H}_{2}\cap\mathcal{H}_{3}$.
    }
    \label{fig:algebra_complete}
\end{figure*}

\begin{proof}

     The intersections $\mathcal{H}_{l^{\prime}}\cap\mathcal{G}_{l}$ and $(\mathcal{H}_{l^{\prime}}\cap\mathcal{H}_{l^{\prime\prime}})\cap\mathcal{G}_{l}$ never contain exactly two elements, since the two elements $G_{1},G_{2}\in\mathcal{H}_{l^{\prime}}\cap\mathcal{G}_{l}$ generate the remaining one element $G_{3}
    =\frac{1}{2i}[G_{1},G_{2}]\in\mathcal{H}_{l^{\prime}}\cap\mathcal{G}_{l}$.
     Therefore, all possible cases with $(\mathcal{H}_{l^{\prime}}\cap\mathcal{H}_{l^{\prime\prime}})\cap\mathcal{G}_{l}\not=\{0\}$ are given by the patterns of (A), (C), and (D).
     
     We next show that the complementary condition $(\mathcal{H}_{l^{\prime}}\cap\mathcal{H}_{l^{\prime\prime}})\cap\mathcal{G}_{l}=\{0\}$ corresponds only to pattern (B). Under the condition $\mathcal{H}_{l^{\prime}}\cap\mathcal{G}_{l}\not=\{0\}$ for any $l^{\prime}\not=l$, at least one element of $\mathcal{G}_{l}$ must be assigned to each of the two intersections, i.e., $G_{1}\in\mathcal{H}_{l^{\prime}}\cap\mathcal{G}_{l}$ and $G_{2}\in\mathcal{H}_{l^{\prime\prime}}\cap\mathcal{G}_{l}$. The remaining element is then $G_{3}$. 

    If we further assume that $G_{3}\in\mathcal{H}_{l^{\prime}}\cap\mathcal{G}_{l}$ in addition to $G_{1}$, it follows that $G_{2}\in\mathcal{H}_{l^{\prime}}$. However, this implies $G_{2}\in(\mathcal{H}_{l^{\prime}}\cap\mathcal{H}_{l^{\prime\prime}})\cap\mathcal{G}_{l}$, which contradicts the assumption. Therefore, we conclude that $G_{3}\not\in\mathcal{H}_{l^{\prime}}$, and consequently $G_{3}\not\in\mathcal{H}_{l^{\prime\prime}}$.
\end{proof}

We construct branches for conditions (i)--(iv) in Theorem~\ref{theorem:2} using Corollary~\ref{corollary:1}.
The algebraic structures for all branches are shown in Fig.~\ref{fig:algebra_complete}.\\

{\it condition (i): $\mathcal{H}_{l}\cap\mathcal{G}_{l^{\prime}}\not=\{0\}$ for any pairs of $l\not=l^{\prime}$}

This condition gives rise to eight branches. 
Four of these eight are classified according to patterns (A) and (B): (i-1) pattern (A) applies to all $\mathcal{G}_{l}$; (i-2) pattern (A) applies to $\mathcal{G}_{1}$ and $\mathcal{G}_{2}$ while pattern (B) applies to $\mathcal{G}_{3}$; (i-3) pattern (B) applies to $\mathcal{G}_{1}$ and $\mathcal{G}_{2}$ while pattern (A) applies to $\mathcal{G}_{3}$; (i-4) pattern (B) applies to all $\mathcal{G}_{l}$. Symmetry under the exchange of $l^{\prime}$ and $l^{\prime\prime}$ implies that permutations of these indices need not be counted separately.

Moreover, the conditions $\mathcal{G}_{2}\cap\mathcal{H}_{3}=\mathcal{G}_{2}$ and $\mathcal{G}_{3}\cap\mathcal{H}_{2}=\mathcal{G}_{3}$ in branches (i-1) and (i-2), respectively, give rise to two additional branches, (i-5) and (i-6) in Fig.~\ref{fig:algebra_complete}. This procedure corresponds to replacing pattern (A) with pattern (D).

The remaining two branches, involving only patterns (C) and (D), are slightly subtler. 
Assigning $\mathcal{G}_{3}$ to pattern (C) yields only two branches: (i-7), in which both $\mathcal{G}_{1}$ and $\mathcal{G}_{2}$ follow pattern (D); and (i-8), in which both $\mathcal{G}_{1}$ and $\mathcal{G}_{2}$ follow pattern (C), together with the additional conditions $\mathcal{H}_{3}\cap\mathcal{G}_{1}=\mathcal{G}_{1}$ and $\mathcal{H}_{3}\cap\mathcal{G}_{2}=\mathcal{G}_{2}$.\\

{\it condition (ii): $\mathcal{H}_{l}\cap\mathcal{G}_{l^{\prime}}\not=\{0\}$ for $\{l,l^{\prime}\}=\{1,2\}$ and $\mathcal{H}_{l}\cap\mathcal{G}_{l^{\prime}}=\{0\}$ otherwise.}

By applying Lemma~\ref{lemma:3} to the subspaces $\mathcal{G}_{3}$, $\mathcal{H}_{1}$, and $\mathcal{H}_{2}$, both the subspaces $\mathcal{G}_{1}$ and $\mathcal{G}_{2}$ trivially intersect with $\mathcal{H}_{3}$. Hence, we obtain two branches: (ii-1) $\mathcal{H}_{1}\cap\mathcal{G}_{2}$ and $\mathcal{H}_{2}\cap\mathcal{G}_{1}$ each contain exactly one element; (ii-2)$\mathcal{H}_{1}\cap\mathcal{G}_{2}=\mathcal{G}_{2}$ and $\mathcal{H}_{2}\cap\mathcal{G}_{1}=\mathcal{G}_{1}$. \\

{\it condition (iii): $\mathcal{H}_{l}\cap\mathcal{G}_{l^{\prime}}=\{0\}$ for $\{l,l^{\prime}\}=\{1,2\}$ and $\mathcal{H}_{l}\cap\mathcal{G}_{l^{\prime}}\not=\{0\}$ otherwise.}

Based on patterns (A) and (B) for $l=3$, we obtain two branches, (iii-1) and (iii-2), in which both the intersections $\mathcal{H}_{3}\cap\mathcal{G}_{1}$ and $\mathcal{H}_{3}\cap\mathcal{G}_{2}$ contain exactly one element. Pattern (C) contradicts the algebraic structure of the intersection $\mathcal{H}_{1}\cap\mathcal{H}_{2}\cong \mathcal{T}^{3}$ stemming from this condition using Lemma~\ref{lemma:3}.

Under pattern (D), where $\mathcal{H}_{1} \cap \mathcal{G}_{3} = \mathcal{G}_{3}$ and the intersection $\mathcal{H}_{2} \cap \mathcal{G}_{3}$ contains one element, we obtain the remaining third branch (iii-3), namely $\mathcal{G}_{1} \cap \mathcal{H}_{3}=\mathcal{G}_{1}$. The intersection $\mathcal{G}_{2} \cap \mathcal{H}_{3}$ contains exactly one element since the condition $\mathcal{G}_{2} \cap \mathcal{H}_{3} = \mathcal{G}_{2}$ leads to $(\mathcal{H}_{1} \cap \mathcal{H}_{2})\cap\mathcal{G}_{3} = \mathcal{G}_{3}$, which contradicts $\mathcal{H}_{1} \cap \mathcal{H}_{2}\cong \mathcal{T}^{3}$.

Moreover, there exists an element of the intersection $\mathcal{H}_{1} \cap \mathcal{H}_{2} \cap \mathcal{H}_{3}$. As an example, we choose two bases $\mathcal{G}_{1} = \mathrm{span}_{\mathbb{R}}\{X_{1}, Y_{1}, Z_{1}\}$ and $\mathcal{G}_{2} = \mathrm{span}_{\mathbb{R}}\{X_{1}X_{2}X_{3}, Y_{1}Y_{2}Y_{3}, Z_{1}Z_{2}Z_{3}\}$, and then we obtain $\mathcal{H}_{1} \cap \mathcal{H}_{2} = \mathrm{span}_{\mathbb{R}}\{X_{2}X_{3}, Y_{2}Y_{3}, Z_{2}Z_{3}\}$. Next, let $\mathcal{G}_{3} = \mathrm{span}_{\mathbb{R}}\{X_{2}X_{3}, Y_{2}X_{3}, Z_{2}\}$, satisfying pattern (D) as $\mathcal{H}_{1} \cap \mathcal{G}_{3} = \mathcal{G}_{3}$ and $\mathcal{H}_{2} \cap \mathcal{G}_{3} = \mathrm{span}_{\mathbb{R}}\{X_{2}X_{3}\}$. Since the operator $Z_{2}Z_{3}$ commutes with the subspace $\mathcal{G}_{3}$, we have $Z_{2}Z_{3} \in \mathcal{H}_{3}$, corresponding to the element of the intersection $\mathcal{H}_{1} \cap \mathcal{H}_{2} \cap \mathcal{H}_{3}$.\\

{\it condition (iv): $\mathcal{H}_{l}\cap\mathcal{G}_{l^{\prime}}=\{0\}$ for any pairs of $l\neq l'$.}

In addition to the algebraic structure explained in Theorem~\ref{theorem:classification_of_SU(8)_by_SU(2)} corresponding to branch (iv-1), we further obtain two branches. The first and simplest one, (iv-2), corresponds to the case in which the intersection among all three subspaces $\mathcal{H}_{l}$ is trivial. The other one, (iv-3), corresponds to the case in which the intersection contains only one element.

\section{Specific examples of Trotter decomposition and circuit construction\label{appendix:examples}}

This section focuses on the concrete decomposition procedures for the physical models listed in Table~\ref{tab:list_of_Hamiltonians}.

\subsection{1D transverse-field Ising model}

We consider the propagator generated by the following Hamiltonian 
\begin{equation}
\begin{split}
    H = J\sum_{i=1}^{N} Z_{i}Z_{i+1} - h\sum_{i=1}^{N}X_{i},
\end{split}
\end{equation}
with the periodic boundary conditions (PBC) of $Z_{1}=Z_{N+1}$.
In the conventional decomposition process, the propagator generated by $H$ is decomposed into two tractable Trotter blocks generated by $H_{1}=-h\sum_{i=1}^{N}X_{i}$ and $H_{2}=J\sum_{i=1}^{N}Z_{i}Z_{i+1}$. Due to the Trotter approximation, the non-commuting pairs between $H_{1}$ and $H_{2}$ appear as
\begin{equation}
[H_{1},H_{2}]=2hJ\sum_{i=1}^{N}(Z_{i}Y_{i+1}+Y_{i}Z_{i+1}).\label{eq:Ising_commut}
\end{equation}
Hence, the number of error terms per single vertex in the conventional decomposition is two.

Our proposed method alternatively decomposes $H$ into two operators $H_{j}$ for $j=1$ and $2$ as
\begin{equation}
\begin{split}
    H_{j}&=J\sum_{i=0}^{N/4-1}\bigl(Z_{4i+2j-1}Z_{4i+2j}+Z_{4i+2j}Z_{4i+2j+1}\bigr)\\
    &\quad-h\sum_{i=0}^{N/4-1}\bigl(X_{4i+2j-1}+X_{4i+2j}\bigr).
\end{split}
\end{equation}
Here, we focus on the three-qubit cluster Hamiltonian $H_{\mathrm{Ising}}^{(3)}$ defined as
\begin{equation}
    H_{\mathrm{Ising}}^{(3)}=J(Z_{1}Z_{2}+Z_{2}Z_{3})-h(X_{1}+X_{2}).
\end{equation}
Since $H_{\mathrm{Ising}}^{(3)}$ satisfies $\mathrm{SU}(2)$ symmetry and commutes with all the generators in $\mathcal{G}_{1}=\mathrm{span}_{\mathbb{R}}\bigl\{X_{1}X_{2}X_{3},X_{1}X_{2}Y_{3},Z_{3}\bigr\}$, its propagator $\exp(-iH_{\mathrm{Ising}}^{(3)}\Delta t)$ is a tractable $\mathrm{SU}(4)$ operator. The encoding gate $U_{1}=CNOT(3\rightarrow2)\circ CNOT(2\rightarrow1)$ unitarily transforms $\mathcal{G}_{1}$ to $\mathrm{span}_{\mathbb{R}}\bigl\{X_{3},Y_{3},Z_{3}\bigr\}$, and thus we obtain the transformed two-qubit effective Hamiltonian whose propagator is implementable.
The non-commuting pairs in our decomposition are finally given by
\begin{equation}
\begin{split}
    [H_{1},H_{2}]=2hJ\sum_{i=1}^{N/2}(-1)^{i}Z_{2i}Y_{2i+1},
\end{split}
\end{equation}
and therefore the number of residual-error terms per vertex is reduced to 0.5.

\subsection{Two-layer \texorpdfstring{$J_{1}$--$J_{2}$}{J1--J2} Heisenberg model}

The two-layer $J_{1}$--$J_{2}$ chain refers to the 1D Heisenberg chain with next-nearest-neighbour interactions. The Hamiltonian can be described by the two-body zig-zag interaction edges (red and blue bars) in FIG.~\ref{fig:NNNdecomp} (a) with the coupling parameter $J_{1}$ and the two-body horizontal ones (green and yellow bars) with the coupling parameter $J_{2}$. The special case of $J_{2}=0$ corresponds to the 1D nearest-neighbour $XXX$ Heisenberg model, which has already been discussed in our proposed approach in Ref.~\cite{Yang2026Quantum}.

\begin{figure}
    \centering
    \includegraphics[width=0.8\linewidth]{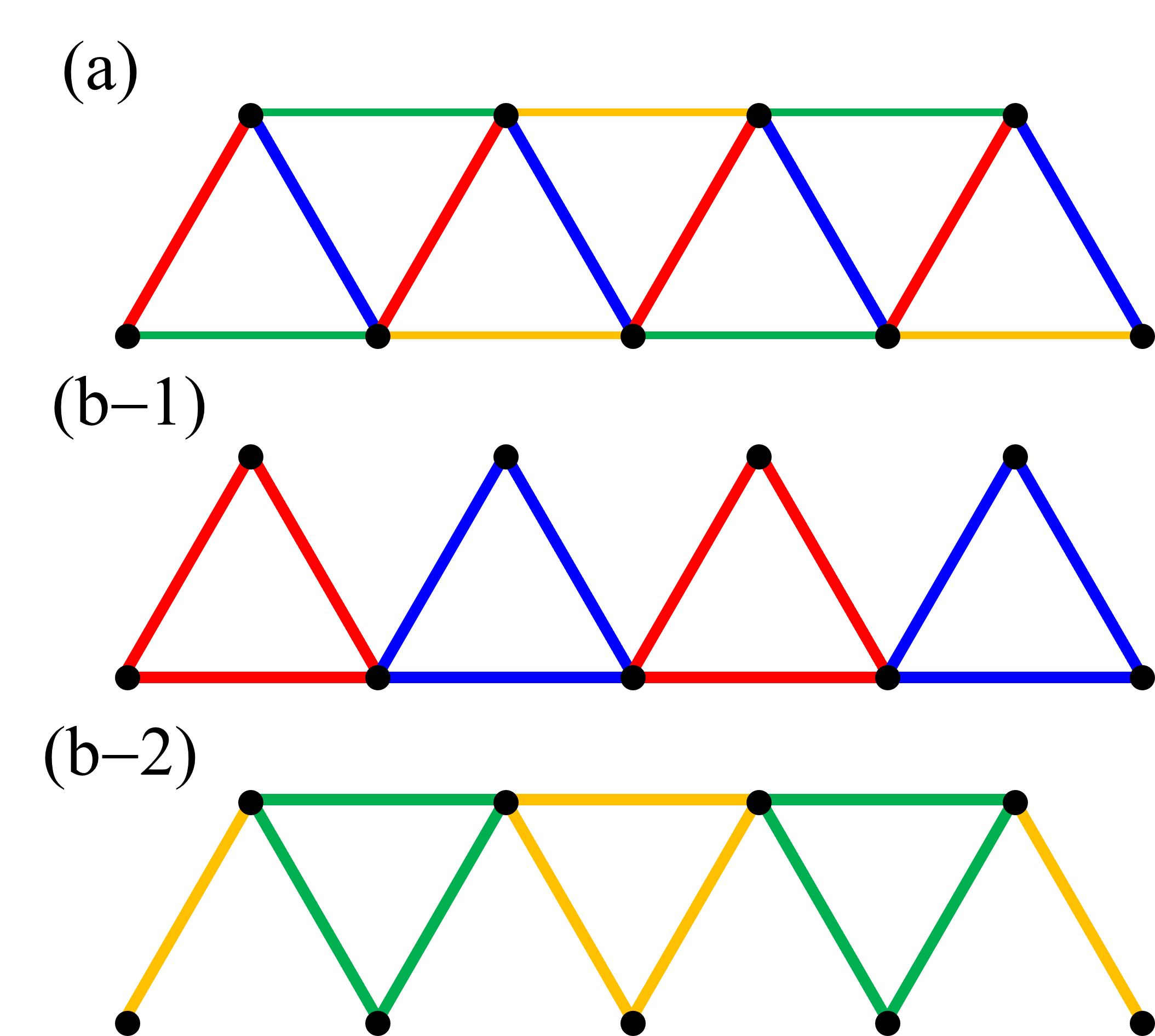}
    \caption{The abstract of the decomposition procedure classified by the coloured edges for (a) the conventional and (b-1,2) the proposed methods.}
    \label{fig:NNNdecomp}
\end{figure}

For the case of the conventional approach, the propagator is approximated as the product of four Trotter blocks generated by the partitioned interactions $J_{1}\vec{\sigma}_{i}\cdot\vec{\sigma}_{j}$ for red and blue bars and $J_{2}\vec{\sigma}_{i}\cdot\vec{\sigma}_{j}$ for green and yellow bars in Fig. ~\ref{fig:NNNdecomp} (a). Six terms in the commutator between two different interaction edges appear as the third line in Eq.~\eqref{eq:epsilon_{1}}
multiplied by the six combinations in choosing two interaction edges among the four partitioned ones in Fig. ~\ref{fig:NNNdecomp} (a). 

We obtain (i) one $J_{1}-J_{1}$ non-commuting edge pair, (ii) four $J_{1}$--$J_{2}$ non-commuting edge pairs, and (iii) one $J_{2}-J_{2}$ non-commuting edge pair. However, the two $J_{1}$--$J_{2}$ non-commuting edge pairs vanish since the sum of two nearest-neighbour edges $J_{1}(\vec{\sigma}_{i-1}\cdot\vec{\sigma}_{i}+\vec{\sigma}_{i}\cdot\vec{\sigma}_{i+1})$ commute with the next-nearest-neighbour edge $J_{2}\vec{\sigma}_{i-1}\cdot\vec{\sigma}_{i+1}$. As a result, we yield the sum of residual non-commuting terms given by (i) 6 $J_{1}-J_{1}$ couplings, (ii) 6$\times2=12$ $J_{1}-J_{2}$ couplings only from the inter-triangular operator pairs, and (iii) 6 $J_{2}-J_{2}$ couplings per vertex.

In our approach, we propose the following Trotter decomposition as
\begin{equation}
\begin{split}
    \exp(-iHt)&=\bigl(\exp(-iH_{4}\Delta t)\exp(-iH_{3}\Delta t)\\\quad&\times\exp(-iH_{2}\Delta t)\exp(-iH_{1}\Delta t)\bigr)^{n/2}\\
    \quad &+O\left(\frac{\epsilon}{n}\right),\label{eq:Trotter_NNN}
\end{split}
\end{equation}
where the power of $n/2$ in Eq.~\eqref{eq:Trotter_NNN} is integer. The approximated propagator is represented by the four Trotter blocks for $2\Delta t$. Thus, our approach uses effectively two Trotter blocks per time step.
The partitioned operators $H_{1}$ and $H_{2}$ correspond to interactions denoted by red and blue upper triangles in Fig.~\ref{fig:NNNdecomp} (b-1) and $H_{3}$ and $H_{4}$ to green and yellow lower counterparts in Fig.~\ref{fig:NNNdecomp} (b-2), respectively. The interactions within the coloured triangles are equivalent to the cluster Hamiltonian given by
\begin{equation}
 H_{J_{1}-J_{2}}^{(3)}=J_{1}\vec{\sigma}_{1}\cdot\vec{\sigma}_{2}+J_{1}\vec{\sigma}_{2}\cdot\vec{\sigma}_{3}+2J_{2}\vec{\sigma}_{1}\cdot\vec{\sigma}_{3}   .
\end{equation}

Our Trotterisation in Eq.~\eqref{eq:Trotter_NNN} significantly reduces the non-trivial commutators $\epsilon$ as follows: 
\begin{equation} 
\begin{split} 
\epsilon&=\frac{[H_{3},H_{4}]+[H_{1},H_{2}]}{2}\\ &+J_{1}J_{2}\left[\sum_{i=1}^{N}\vec{\sigma}_{i}\cdot\vec{\sigma}_{i+1},\sum_{j=1}^{N}(-1)^j\vec{\sigma}_{j}\cdot\vec{\sigma}_{j+2}\right].\label{eq:NNN_error} 
\end{split} 
\end{equation}
The reduction in the second line originates from the term $\frac{1}{2}[H_{1}+H_{2},H_{3}+H_{4}]$,
because the total nearest-neighbour interaction contained in both $H_{1}+H_{2}$ and $H_{3}+H_{4}$ commutes trivially.

We next discuss the residual-error terms in Eq.~\eqref{eq:NNN_error} from an intuitive perspective. 
In our decomposition procedure, the three-qubit propagators associated with the upper triangles are implemented exactly, whereas the propagators associated with the lower triangles are still approximated through Trotterisation. 
Since the numbers of upper and lower triangles are identical in this system, the residual terms originating from intra-triangular non-commutativity are reduced by half compared with the conventional decomposition.
Among the intra-triangular residual terms, only the $J_{1}$--$J_{1}$ couplings contribute. 
Consequently, the number of $J_{1}$--$J_{1}$ residual terms is reduced from six to three per vertex.

By contrast, our decomposition does not reduce the residual-error terms originating from inter-triangular operators, namely the $J_{1}$--$J_{2}$ and $J_{2}$--$J_{2}$ couplings. 
Taking into account the prefactor $1/2$ in Eq.~\eqref{eq:NNN_error} and the factor $2$ in the next-nearest-neighbour interaction contained in $H_{J_{1}-J_{2}}^{(3)}$, the magnitude of the $J_{1}$--$J_{2}$ residual terms remains unchanged from the conventional approach:
$    J_{1}\times 2J_{2} /2 = J_{1}J_{2}$.
On the other hand, the residual contribution from the $J_{2}$--$J_{2}$ couplings becomes twice as large as that in the conventional decomposition:
    $2J_{2}\times 2J_{2} /2 = 2J_{2}^2$.

We present the detailed calculation for counting the residual terms as follows.
Focusing on a vertex on the lower layer in Fig.~\ref{fig:NNNdecomp}, the non-trivial terms in the first line of Eq.~\eqref{eq:NNN_error} that act non-trivially on this vertex originate only from $[H_{1},H_{2}]$ in Fig.~\ref{fig:NNNdecomp}(b-1). 
Taking into account the prefactor $1/2$ in Eq.~\eqref{eq:NNN_error} and the factor $2$ in the next-nearest-neighbour interaction of $H_{J_{1}-J_{2}}^{(3)}$, the first line gives the following residual-error contributions per vertex:
(i) $6/2=3$ $J_{1}$--$J_{1}$ terms,
(ii) $6\times2/2+6\times2/2=12$ $J_{1}$--$J_{2}$ terms,
and
(iii) $6\times2\times2/2=12$ $J_{2}$--$J_{2}$ terms.

The second line of Eq.~\eqref{eq:NNN_error} does not increase the effective number of residual-error terms. The second line also results in the reduction in the commutator pairs between $\vec{\sigma}_{i}\cdot\vec{\sigma}_{i+1}+\vec{\sigma}_{i+1}\cdot\vec{\sigma}_{i+2}$ and $\vec{\sigma}_{i}\cdot\vec{\sigma}_{i+2}$. We obtain the second line simply given by the sum of two non-commutative pairs on the $i$-th vertex as $(-1)^i[\vec{\sigma}_{i}\cdot\vec{\sigma}_{i+1},\vec{\sigma}_{i-2}\cdot\vec{\sigma}_{i}]$ and $(-1)^i[\vec{\sigma}_{i-1}\cdot\vec{\sigma}_{i},\vec{\sigma}_{i}\cdot\vec{\sigma}_{i+2}]$. Taking into account the sign factor $(-1)^i$, one of these commutators cancels a contribution from the $J_{1}$--$J_{2}$ terms in (ii), whereas the other duplicates an existing contribution. 
Therefore, the cancellation and duplication cancel each other out, and the second line does not change the effective number of residual-error terms.

\subsection{2D transverse-field Ising model}

We consider the following Hamiltonian 
\begin{equation}
\begin{split}
    H = J\sum_{i=1}^{N_X}\sum_{j=1}^{N_Y} Z_{i,j}Z_{i+1,j}+Z_{i,j}Z_{i,j+1} - h\sum_{i=1}^{N_X}\sum_{j=1}^{N_Y}X_{i,j},
\end{split}
\end{equation}
where the double subscripts ``$i,j$'' denote $x$- and $y$-coordinates of the 2D lattice site, respectively. The periodic boundary conditions(PBC) are given by $Z_{1,j}=Z_{N_X+1,j}$ and $Z_{i,1}=Z_{i,N_Y+1}$.

We conventionally decompose $H$ into three tractable Trotter blocks generated by $H_{1}=-h\sum_{i=1}^{N_X}\sum_{j=1}^{N_Y}X_{i,j}$, $H_{2}=J\sum_{i=1}^{N_X}\sum_{j=1}^{N_Y}Z_{i,j}Z_{i+1,j}$, and  $H_{3}=J\sum_{i=1}^{N_X}\sum_{j=1}^{N_Y}Z_{i,j}Z_{i,j+1}$. We note that $H_{2}+H_{3}$ is not a tractable form despite the commutativity $[H_{2},H_{3}]=0$ since the propagator generated by the four interaction edges $Z_{i,j}Z_{i,j+1}+Z_{i,j}Z_{i,j-1}+Z_{i+1,j}Z_{i,j}+Z_{i,j}Z_{i-1,j}$ is not SU(4) operator and not implementable within single Trotter block. In the same manner as the commutation relation in Eq.~\eqref{eq:Ising_commut}, the number of error terms per vertex in the conventional decomposition is four.

Our proposed method decomposes $H$ alternatively into two operators $H_{j}$ for $j=1$ and $2$ as
\begin{equation}
\begin{split}
    H_{1}&=J\sum_{i=1}^{N_X}\sum_{j=1}^{N_Y}Z_{i,j}Z_{i+1,j}-h\sum_{i=1}^{N_X}\sum_{j=1}^{N_Y}\frac{1+(-1)^{i+j}}{2}X_{i,j},\\
    H_{2}&=J\sum_{i=1}^{N_X}\sum_{j=1}^{N_Y}Z_{i,j}Z_{i,j+1}-h\sum_{i=1}^{N_X}\sum_{j=1}^{N_Y}\frac{1+(-1)^{i+j+1}}{2}X_{i,j}.
\end{split}
\end{equation}
Our partitioning of the transverse field in the second term of $H_{1}$ and $H_{2}$ is determined by whether the Manhattan distance between the central site $(0,0)$ and the site $(i,j)$ is odd or even.

Both the two partitioned Hamiltonians are spanned by the three-qubit cluster Hamiltonians unitarily equivalent to 
\begin{equation}
    H_{\mathrm{Ising}}^{(3)}=J(Z_{1}Z_{2}+Z_{2}Z_{3})-hX_{2}.
\end{equation}
In the same way as the 1D transverse Ising model, this propagator $\exp(-iH_{\mathrm{Ising}}^{(3)}\Delta t)$ is a tractable $\mathrm{SU}(4)$ operator, indicating that we obtain a two-qubit effective Hamiltonian by the same encoding procedure.

The non-commuting pair on the site $(i,j)$ satisfying $(-1)^{i+j}=1$ in our decomposition are finally given by
\begin{equation}
\begin{split}
    [X_{i,j},H_{2}]=2hJY_{i,j}(Z_{i,j-1}+Z_{i,j+1}),
\end{split}
\end{equation}
and thus the number of error terms per vertex in our decomposition is two.

\subsection{2D triangular lattice}

We focus on the 2D triangular lattice containing two-body $XXX$ Heisenberg and three-body scalar spin-chirality interactions. The other types of interaction model listed in Table~\ref{tab:list_of_Hamiltonians} are the specific cases of the 2D triangular lattice, which can be obtained by truncating some kinds of interactions. 

\begin{figure}
    \centering
    \includegraphics[width=1.0\linewidth]{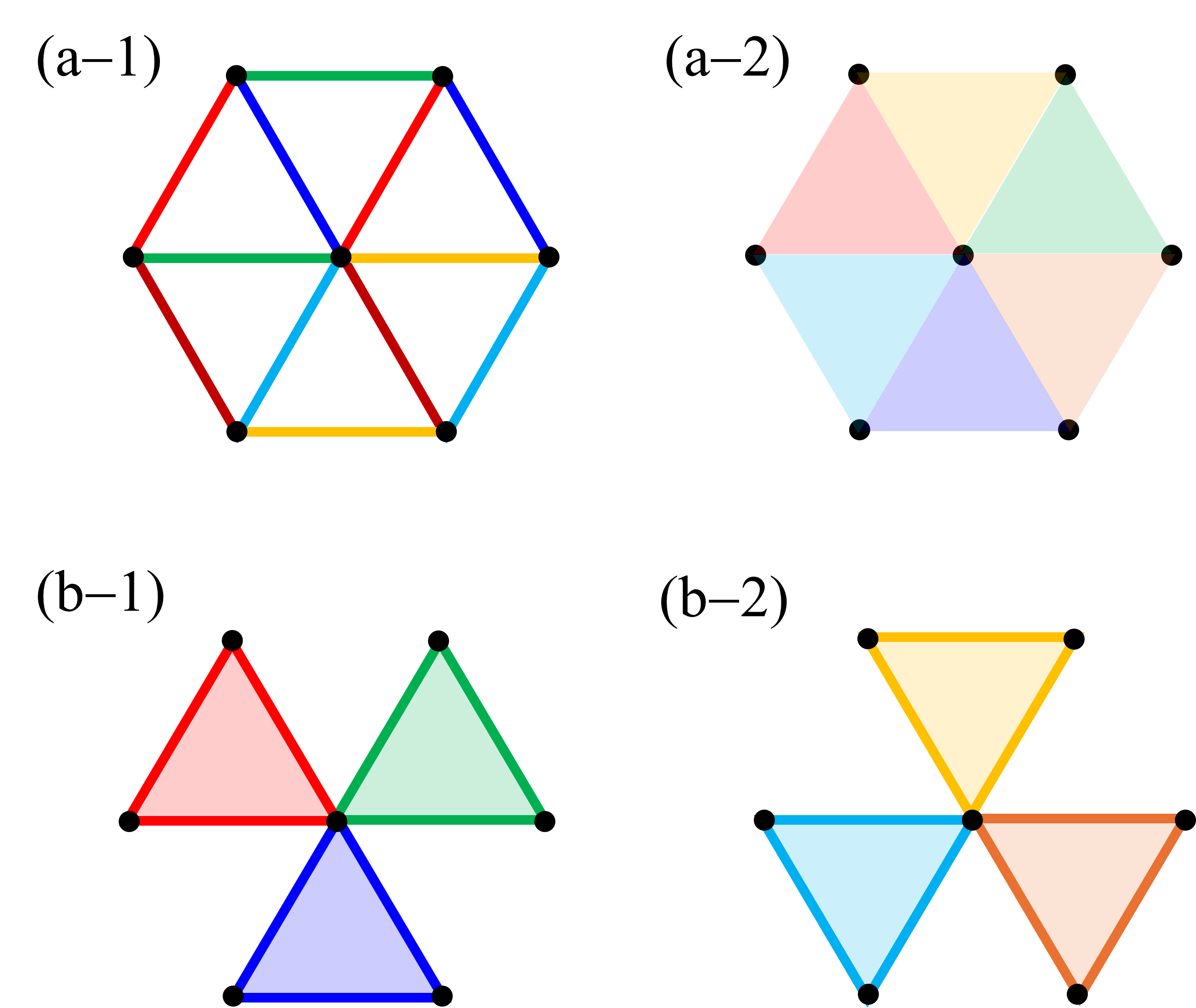}
    \caption{Decomposition procedure classified by the coloured edges and triangles for (a--1,2) the conventional and (b--1,2) the proposed methods.}
    \label{fig:triangledecomp}
\end{figure}

For the case of the conventional approach, the propagator is approximated as the product of four Trotter blocks generated by the partitioned interactions $J\vec{\sigma}_{i}\cdot\vec{\sigma}_{j}$ for coloured bars and $K\vec{\sigma}_{i}\cdot(\vec{\sigma}_{j}\times\vec{\sigma}_{k})$ coloured triangles in Fig.~\ref{fig:triangledecomp} (a). 

We first count the non-commuting terms between two edge interactions. There are two types of non-commuting pairs in this Trotter approximation among the six partitioned interaction edges in Fig.~\ref{fig:triangledecomp} (a--1). The first type corresponds to the pairs of edges within the triangles, of which there are six. In the same manner as the $J_{1}$--$J_{2}$ next-nearest-neighbour model, the number of non-commuting terms is six per triangle. This implies that the non-commuting pairs within the triangles amount to $6\times6/3=12$ per vertex; the division by three reflects the number of vertices in one triangle. The second type corresponds to the nine edge-pairs outside the triangles, implying that there are $6\times9=54$ non-commuting pairs per vertex.

For the non-trivial commutator pairs between two triangles, there are three patterns to be taken into account: non-commutative pairs in (I) internal triangular interaction, (II) 
between different triangles sharing an edge, and (III) between different triangles sharing only the vertex. For pattern (I), the non-commutative pairs in a triangle amount to 6 per three vertices, as explained in Eq.~\eqref{eq:epsilon_{1}}. The number of residual terms per vertex is then equal to $6\times6/3=12$. For pattern (II), there are six pairs of two triangles sharing an edge. For a pair, there are 12 non-commuting pairs per two vertices, so that we obtain $6\times12/2=36$ terms in pattern (II). For pattern (III), there are nine pairs of two triangles sharing only a vertex. For a pair, there are 24 non-commuting pairs, yielding $9\times24=216$ terms in pattern (III). The total number of the non-commuting terms in inter-triangular pairs is equal to $36+216=252$.

There also exist the non-commutative pairs between an edge interaction and a triangular interaction: the triangle covering the edge and the triangle and the edge sharing only a vertex. For the triangle covering the edge, however, the sum of three interaction edges within the triangle commutes with spin-chirality interactions. Therefore, we only need to count the pairs of a triangle and an edge sharing only a vertex. There are 24 pairs of an edge and a triangle that share only one vertex. For a pair, there are 12 non-commuting pairs, so that we obtain $12\times24=288$ terms.

In our approach, we approximate the propagator as follows:
\begin{equation}
\begin{split}
    &\exp(-iHt)\\
    &=\bigl(\exp(-iH_{6}\Delta t)\exp(-iH_{5}\Delta t)\exp(-iH_{4}\Delta t)\\\quad&\times\exp(-iH_{3}\Delta t)\exp(-iH_{2}\Delta t)\exp(-iH_{1}\Delta t)\bigr)^{n/2}\\
    \quad &+O\left(\frac{\epsilon}{n}\right).\label{eq:Trotter_triangle}
\end{split}
\end{equation}
Since the approximated propagator consists of six Trotter blocks for $2\Delta t$, our approach uses essentially three Trotter blocks per time step.
The partitioned operators $H_{1}$, $H_{2}$, and $H_{3}$ correspond to interactions denoted by red, green, and blue upper triangles in Fig.~\ref{fig:triangledecomp} (b-1) and $H_{4}$, $H_5$, and $H_6$ to yellow, orange, and light blue lower counterparts in Fig.~\ref{fig:triangledecomp} (b-2), respectively. The partitioned Hamiltonians $H_m$ contain three-site interactions equivalent to 
\begin{equation}
\begin{split}
    H_{\mathrm{triangle}}^{(3)}&=J\bigl(\vec{\sigma}_{1}\cdot\vec{\sigma}_{2}+\vec{\sigma}_{2}\cdot\vec{\sigma}_{3}+\vec{\sigma}_{1}\cdot\vec{\sigma}_{3}\bigr)\\
&\quad+2K\vec{\sigma}_{1}\cdot\left(\vec{\sigma}_{2}\times\vec{\sigma}_{3}\right).    
\end{split}
\end{equation}

In our Trotter approximation, we can reduce the residual error from non-commuting pairs within the triangle. The number of vanishing terms amounts to six per triangle, which corresponds to $6\times3/3=6$ per vertex, and thus the number of non-commuting terms is $12-6=6$. For the non-trivial commutator pairs between two triangles, we also reduce the residual error stemming from internal triangular interactions, which have 12 non-commuting terms per vertex. 

On the other hand, the non-trivial commutator pairs in the other patterns do not vanish. Considering the double prefactor of scalar spin-chirality and half prefactor of Trotterisation, the total number of the non-commuting terms in triangular interactions is equal to $(36+216)\times2^{2}/ 2=504$.

The non-commutative pairs between an edge interaction and a triangular interaction are also considered. Taking into account the aforementioned prefactors, the total number of the non-commuting terms is equal to 288 in the same manner as the conventional decomposition.

\begin{figure}
    \centering
    \includegraphics[width=1.0\linewidth]{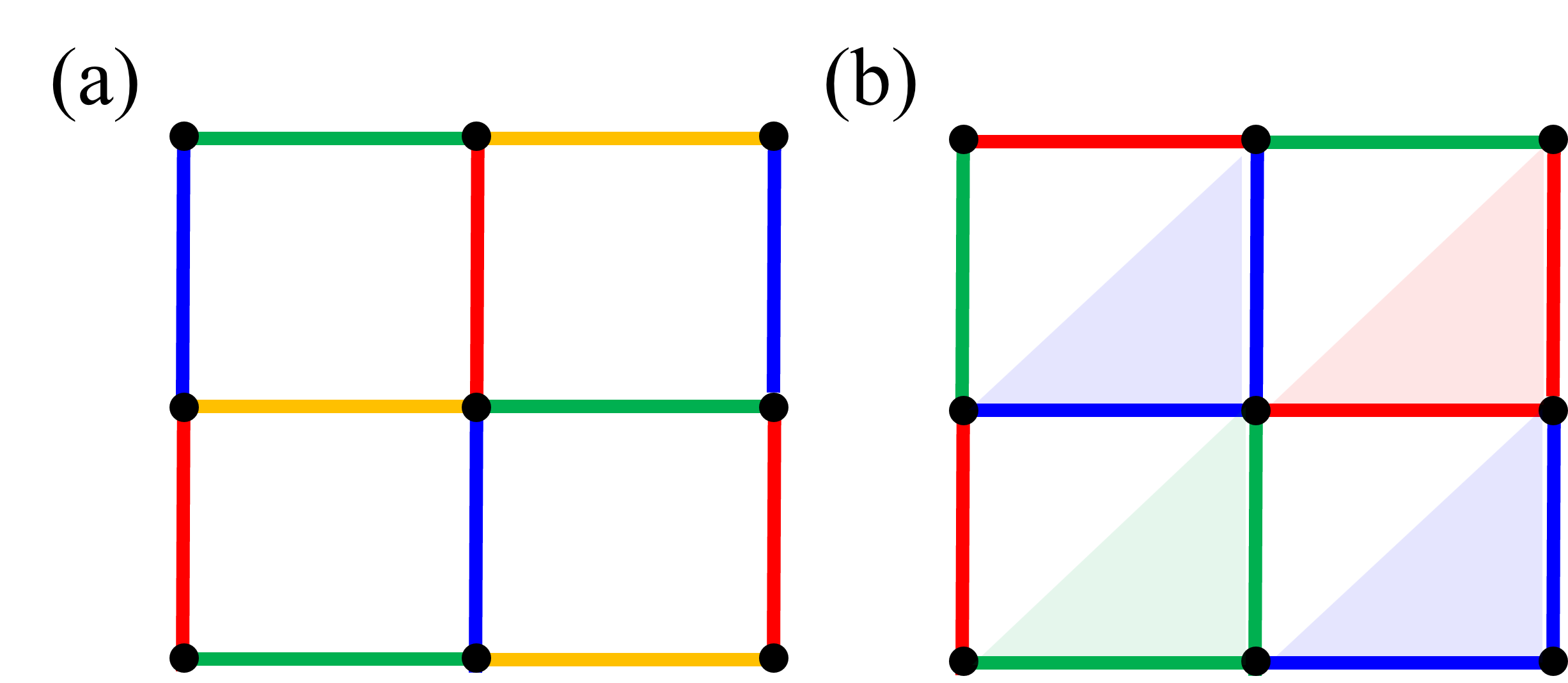}
    \caption{ Decomposition procedures classified by the coloured edges for (a) the conventional and (b) the proposed methods in the square lattice.
    }
    \label{fig:squaredecomp}
\end{figure}

\subsection{2D square Heisenberg lattice}

We finally remark that both the conventional and the proposed methods can be applied to the two-dimensional square $XXX$ Heisenberg lattice by removing the blue and orange edges from the triangular lattice shown in Fig.~\ref{fig:triangledecomp}(a-1).
Here, we illustrate the partitioned interaction edges in Fig.~\ref{fig:squaredecomp}.

While the conventional decomposition requires four Trotter blocks to represent the full Hamiltonian, as shown in Fig.~\ref{fig:squaredecomp}(a), our decomposition reduces the number of Trotter blocks by one, as illustrated in Fig.~\ref{fig:squaredecomp}(b).


\bibliography{main}


\end{document}